\documentclass[a4paper]{article}

\usepackage{a4wide}

\usepackage{amsmath}
\usepackage{graphicx}
\usepackage{amsfonts}
\usepackage{amssymb}
\usepackage{bm}
\usepackage{algorithm}
\usepackage{algorithmic}
\usepackage[usenames]{color}
\usepackage{amsthm}
\usepackage{booktabs}   
\usepackage{hyperref}

\RequirePackage{natbib}

\theoremstyle{plain}
\newtheorem{theorem}{Theorem}

\newcommand{\fbar}{{\bar{f}}}
\newcommand{\fh}{{\hat{f}}}
\newcommand{\ft}{{\tilde{f}}}
\newcommand{\xt}{{\tilde{x}}}
\newcommand{\ud}{\, \mathrm{d}}

\newcommand{\Cbar}{{\bar C}} 

\newcommand{\Ccal}{{\mathcal C}} 
\newcommand{\Ical}{{\mathcal I}} 
\newcommand{\Rbb}{{\mathbb R}} 
\newcommand{\Real}{{\mathbb R}} 
\renewcommand{\emptyset}{\varnothing} 
\newcommand{\norm}[1]{\lVert#1\rVert}
\newcommand{\argmin}{\mathop{\mathrm{arg\,min}}}
\newcommand{\argmax}{\mathop{\mathrm{arg\,max}}}
\newcommand{\Gh}{{\widehat{G}}}
\newcommand{\Bh}{{\widehat{B}}}

\title{Merging Modal Clusters via Significance Assessment}

\author{
  Yong Wang\thanks{Department of Statistics, The University of Auckland, New Zealand, email: yongwang@auckland.ac.nz}
  \and
  Shengwei Hu\thanks{Department of Statistics, The University of Auckland, New Zealand, email: shu454@aucklanduni.ac.nz}
}

\date{September 2, 2026}

\begin{document}

\maketitle

\begin{abstract}
  To deal with superfluous clusters and to reduce the number
  of clusters as often desired in practice, a modal cluster merging
  procedure is proposed. Based on some new properties established in
  this paper for Morse functions, the procedure merges clusters in a
  sequential manner without causing unnecessary density
  distortion. Each cluster is evaluated for its significance relative
  to the other clusters, using the Kullback-Leibler divergence or its
  log-likelihood approximation, by truncating the density for the
  cluster at an appropriate level. The least significant cluster is
  then merged into one of its adjacent clusters, using the novel
  concept of cluster adjacency defined in this paper. The resulting
  hierarchical clustering tree is useful for determining the number of
  clusters, as may be preferred by a specific user or in a general,
  meaningful manner. Numerical studies show that the new procedure
  deals well with difficult clustering problems and often produces
  intuitively appealing and numerically more accurate clustering
  results, as compared with several other popular clustering methods
  in the literature.
\end{abstract}

\noindent\textbf{Keywords:} Modal clustering, cluster merging, cluster
significance, likelihood ratio test, Morse function, Hierarchical
clustering

\section{Introduction}
\label{sec:intro}

The mode-based or modal clustering approach has received increasing
research attention in recent years. Being density-based, it identifies
a cluster by the basin of attraction of a mode of a probability
density function. It has more robust performance with data scaling and
outlying observations than traditional distance-based clustering
methods which measure similarity between objects based on some
distance metrics \citep{macqueen-1967,gower-ross-1969,defays-1977}. It
is also able to produce various yet intuitively appealing shapes of
clusters and is thus advantageous over model-based clustering which
uses mixture components, such as Gaussian distributions, to represent
clusters \citep{fraley-raftery-2002}. It also provides a clear
population target—namely, the basins of attraction of the true
density's modes—at which sample-based clustering algorithms should aim
\citep{chacon-2015}. 

Modal clustering methods have two major strands: level-set-based vs.\@
mode-seeking. The earliest idea of level-set-based clustering can date
back to \citet{hartigan-1975} which defines clusters as connected
high-density components separated by low-density boundary regions. By
varying the level which identifies the connected components, a
hierarchical structure of clusters known as cluster tree can therefore
be formed; see, e.g., \cite{cuevas-febrero-fraiman-2001},
\cite{azzalini-torelli-2007}, \cite{stuetzle-nugent-2010},
\cite{cadre-pelletier-pudlo-2013}, \cite{menardi-azzalini-2014} and
the references therein. The mode-seeking method, as its name suggests,
aims at identifying the modes of the underlying density function,
followed by associating each observation with a specific mode. Most of
the modal clustering methods that follow the mode-seeking strand rely
on the mean-shift algorithm proposed by \citet{fukunaga-hostetler-1975} or its
variants. \citet{menardi-2016} gives a review on modal clustering
methods.

For modal clustering, one major difficulty is how to tackle spurious
modes that occur frequently in density estimation. Varying the
smoothness for density estimation is a natural choice to deal with
them. \citet{minnotte-scott-1992} described, without implementing it, a new
concept of ``mode tree'' based on the fact that the number of modes of
a kernel-based density estimator (KDE) decreases as the bandwidth
increases. In a similar spirit, \citet{li-ray-lindsay-2007} proposed the
hierarchical mode association clustering (HMAC) algorithm, which
produces a hierarchical clustering tree by gradually increasing the
bandwidth.  A drawback of HMAC is that the increase of a global
bandwidth may result in severe distortion in density estimates and
thus produce inappropriate clustering outcomes. To overcome this
problem, \citet{hu-wang-2021} proposed a modal clustering method MDE-MF. It
makes use of a mode-flattening (MF) technique, which locally, rather
than globally, modifies the smoothness of density estimates. They also
used a mixture-based density estimator (MDE) to help improve
estimation accuracy and higher-dimensional applicability. However,
though to a lesser extent, the increase of a local bandwidth by MDE-MF
also distorts the density and can thus fail to produce satisfactory
clustering results for harder clustering problems, which will be
further demonstrated in Section~\ref{sec:numeric}.

In this paper, we propose a new modal clustering method that works
with any given density function that is a Morse function. Based on the
new properties of Morse density functions that are established in this
paper, the method avoids the increase of bandwidth value(s) and does
not suffer from the above-mentioned density distortion problem. It
gives more satisfactory results for harder clustering problems and
higher clustering accuracy for common problems. The method is guided
by two novel ideas presented in this paper: cluster adjacency and
significance assessment. As will be detailed later, all clusters are
connected through adjacent clusters and a cluster is only merged into
one of its adjacent clusters at each step, while the significance of a
cluster is assessed by its excess mass, using the Kullback-Leibler
divergence or its log-likelihood approximation. To assess the excess
mass is to truncate the density bump for a cluster at an appropriate
level, but the excess mass is not filled back elsewhere, thus avoiding
density distortion anywhere else. The process continues and creates a
sequence of sub-densities until one large cluster covering the entire
support of the density is obtained. The method is a hierarchical one,
and a dendrogram can be produced. A user can choose the number of
clusters from the clustering tree for further applications.

Throughout the paper, we use usual notations for set operations. For
any $\epsilon > 0$, the $\epsilon$-ball centred at $x \in \Real^d$ is
defined as
\[ 
  S_\epsilon(x) = \{y \in \Real^d: \norm{y - x} < \epsilon\},
\]
where $\norm{\cdot}$ denotes the Euclidean norm. The boundary of a set
$A$ is denoted by $\partial A$, i.e.,
\begin{align*}
  \partial A & = \{x \in \Rbb^d: S_\epsilon(x) \cap A \ne \emptyset
               \mbox{~and~} S_\epsilon(x) \cap A^c \ne \emptyset),
               \mbox{~for~any~}\epsilon > 0\}, 
\end{align*}
where $A^c$ denotes the complement of $A$ in $\Real^d$. Further, $A$
has interior $A^\circ$ and closure $\bar A$, and $A \backslash B$
indicates the complement of $B$ in $A$.

The rest of this paper is organized as follows.
Section~\ref{sec:morsefunctions} briefly reviews modal clustering with
Morse functions.  Section~\ref{sec:connected-clusters} proposes the
concept of cluster adjacency and connectedness and establishes some
new properties that are critical for developing our new clustering
method.  Section~\ref{sec:method} discusses the proposed modal cluster
merging method in full detail. Section~\ref{sec:implementation}
describes the key steps and some practical issues for implementing the
method. In Section~\ref{sec:numeric}, we compare the performance of
the new method and several other popular clustering methods with both
simulated and real-world datasets. Finally, a summary and some remarks
are given in Section~\ref{sec:conclusion}.

\section{Modal Clustering with Morse Functions}
\label{sec:morsefunctions}

This section briefly reviews modal clustering using Morse functions
and establishes the notation used to lay the foundation for this work.
For any density function $f$ that is used throughout this paper for
modal clustering, we assume:
\begin{quote}
  (A) $f$ has support $\Real^d$, is a Morse function and has finitely
  many critical points.
\end{quote}
To remain focused and avoid tedious technical complications, we will
only consider the support $\Real^d$, which is mostly assumed in the
literature. We realize that more general types of support can be used
for our method, and will give a brief discussion in
Section~\ref{sec:conclusion}.

A Morse function has only non-degenerate critical points, i.e., the
Hessian at any critical point is non-singular. As a result, critical
points are isolated and one can find local coordinates
$x_{(1)}, \dots, x_{(d)}$ such that $f$ can be written as
\begin{align}
  \label{eqn:morse-lemma}
  f(x) = f(x^*) - x_{(1)}^2 - \dots - x_{(\lambda)}^2 + x_{(\lambda+1)}^2 +
  \dots + x_{(d)}^2
\end{align}
around a critical point $x^*$, where $\lambda = \lambda(x^*)$, the
number of the negative signs in (\ref{eqn:morse-lemma}) or indeed the
number of the negative eigenvalues of the Hessian matrix at $x^*$, is
known as the Morse index of the critical point. We refer the reader to
\cite{milnor-1963}, \cite{matsumoto-2002},
\cite{banyaga-hurtubise-2004}, \cite{lee-2012} and \cite{jost-2017}
for extensive coverage of the Morse theory and the relevant theory of
smooth manifolds.  Owing to the unique and interesting properties of
Morse functions, they are widely considered for modal clustering. In
aid of Morse theory, \citet{azizyan-chen-etal-2015z} developed a bound
on the Rand index of modal clustering and showed that the clustering
risk is low when clusters are well-separated, even in high-dimensional
scenarios. \citet{chen-genovese-wasserman-2016} proposed several
enhancements to modal clustering. For more works on Morse theory and
modal clustering, see, e.g., \cite{chacon-2015},
\cite{ariascastro-mason-pelletier-2016},
\cite{chen-genovese-wasserman-2017}, \cite{chacon-2019},
\cite{kejzar-korenjakcerne-batagelj-2021} and the references therein.

Now let us consider a density function $f$ that satisfies condition
(A).  Each of its critical points can be a maximum (mode), a minimum
or a saddle point, and any non-critical point is also known as a
regular point. The (gradient) flow starting at a point
$x \in \Real^d$, is defined as the path
$\pi: \mathbb{R} \rightarrow \mathbb{R}^d$ such that
\begin{align*} 
  \pi(0; x, f) & = x, \\
  \pi'(t; x, f) & = \triangledown f(\pi(t; x, f)),
\end{align*}
where $\triangledown f$ denotes the gradient vector of $f$ and thus
$\pi'(t; x, f)$ the direction of the steepest ascent at
$\pi(t; x, f)$. The entire curve $\pi(t; x, f)$, $t \in \Real$, is
known as a flow line, and for simplicity we denote the flow line
starting at $x$ by
\begin{align*}
  \pi(\Real; x , f) = \{\pi(t; x, f): t \in \Real\}.
\end{align*}
It is known that starting at any other point on a flow line
$\pi(\Real; x, f)$ creates the same flow line, and that two flow lines
are either disjoint or identical. The destination of a flow line is
defined as
\[
  \pi(\infty; x, f) = \lim_{t \to \infty}\pi(t; x, f),
\]
which must be a critical point. A flow line is either constant (i.e.,
at a critical point) or strictly increasing (i.e., at a regular
point), as $t$ increases. Also, we denote the set of destinations of
the flow lines starting at each point in a set $\mathcal{X}$ as
\[
  \pi(\infty; \mathcal{X} , f) = \{\pi(\infty; x, f):  x \in \mathcal{X}\}.
\]
The basin of attraction of a critical point $x^*$ is defined as
\[ 
  C(x^*) = \{x \in \Real^d: \pi(\infty; x, f) = x^*\},
\]
which is a connected manifold of dimension $\lambda(x^*)$. This is
also known as the stable manifold at $x^*$ of the flow, while
\[ 
  C^u(x^*) = \{x \in \Real^d: \pi(-\infty; x, f) = x^*\}
\]
is the unstable manifold of dimension $d - \lambda(x^*)$ at $x^*$. The
sample space $\Rbb^d$ is the disjoint union of manifolds $C(x^*)$ of
dimensions ranging from $0$ to $d$, over all critical points $x^*$.

For a density $f$, let us denote by $M_j$ its set of critical points
of Morse index $j$, e.g., $M_d$ denotes the set of the modes and
$M_{d-1}$ the set of the critical points of Morse index $d-1$.  For
any mode $x^* \in M_d$, $C(x^*)$ is a $d$-manifold and defines a modal
cluster, and the union set $\cup_{x^* \in \cup_{j < d} M_j} C(x^*)$
forms the boundaries of the modal clusters. Note that the boundary
points are undetermined for clustering purposes. This causes no
practical problem, as the set is of Lebesgue measure $0$. One may
arbitrarily allocate the boundary points to any of its bordering
clusters.

\begin{figure}[!tbh]  
  \centering
  \includegraphics[width=0.9\textwidth]{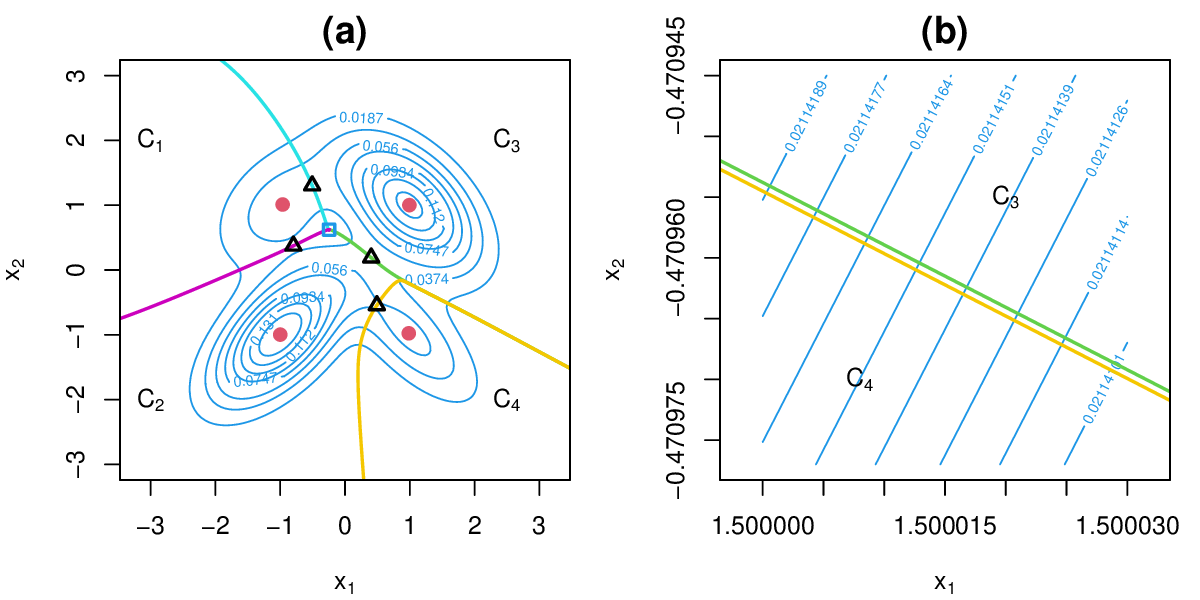}
  \caption{(a) A density with four modes; (b) A blowup of a small
    area} 
  \label{fig:bound}
\end{figure}

As a simple illustration, Fig.~\ref{fig:bound}(a) shows how modal
clustering works, using the true density of the Quadrimodal
distribution described in \citet{wand-jones-1993}. As shown in the
plot, the density function has $4$ modes (solid points), and $4$
saddle points (triangles) and a minimum (square). The curves (in
different colours) consist of the flow lines whose destinations are
non-mode critical points. These curves form the boundaries of the $4$
modal clusters, labelled $C_1$--$C_4$. It is interesting to note that
the boundaries of $C_3$ and $C_4$ never intersect, no matter how
seemingly close they are to each other for $x_1$ large, as shown in
Fig.~\ref{fig:bound}(b). The narrow strip between the two boundary
curves shown in Fig.~\ref{fig:bound}(b) belongs to cluster $C_2$, even
though it is closer in Euclidean distance to the modes of $C_3$ and
$C_4$.

\section{Cluster Adjacency and Connectedness}
\label{sec:connected-clusters}

In this section, we define some new concepts and establish some
relevant theoretical properties about modal clusters, which are
critically important for developing the new method described in
Section~\ref{sec:method}. Also, we may simply say a cluster to mean a
modal cluster in this section; in later sections, a cluster may also
mean a union of modal clusters.

While one may simply define that two clusters are adjacent if their
boundaries have a non-empty intersection, we can actually choose a
more restrictive, yet equivalent definition. We say that two different
clusters $C_i = C(x_i^*)$ and $C_j = C(x_j^*)$,
$x_i^*, x_j^* \in M_d$, are \emph{adjacent}, if
$\partial C_i \cap \partial C_j$ contains any non-empty subset of the
stable manifold of a critical point of Morse index $d-1$. As will be
seen, this will lead to the use of the critical points in $M_{d-1}$
only. These critical points tend to have higher density values, and
using them only is sufficient and gives the same performance for our
cluster merging algorithm. We hence only need to consider a smaller
number of critical points, rather than all critical points. The
largest union of the subsets of the stable $(d-1)$-manifolds contained
in $B_{ij} \equiv \partial C_i \cap \partial C_j$ is said to be the
\emph{border} between $C_i$ and $C_j$.  For adjacent clusters, we may
also say that one cluster is an (adjacent) neighbour of the other,
while borderless clusters are not neighbours. Adjacent clusters are in
close vicinity to each other and can potentially be merged, and
clusters without borders can not be merged into each other directly.

The cluster merging algorithm to be presented later requires the
following theoretical results. Note that condition (A) is always
assumed.

\begin{theorem}
  \label{thm:boundary-flow-line}
  $\pi(\Real; x , f) \subseteq \partial C(x^*)$ if
  $x \in \partial C(x^*)$, $x^* \in M_d$.
\end{theorem}

\begin{proof}
    
  It is obvious, if $x$ is a critical point.

  Now let $x$ be a regular point, and write $C = C(x^*)$. We only
  focus on $t > 0$, since the proof is analogous for $t < 0$. Assume
  that $\pi(t; x, f) \notin \partial C$ for some $0 < t < \infty$.
  Let
  \begin{align*}
    \tilde t = \inf \{t \in [0,\infty): \; \pi(t; x, f) \notin
    \partial C\},
  \end{align*}
  and denote $\tilde{x} = \pi(\tilde{t}; x, f)$.  Clearly,
  $\tilde{x} \in \partial C$. Consider $S_\epsilon(\xt)$, for a
  sufficiently small $\epsilon > 0$.  Since $\partial C$ contains the
  flow line $\pi(t; x, f)$ all the way up to $\tilde{t}$,
  $S_\epsilon(\xt)$ contains a $d$-dimensional open subset of $C$ and
  a $d$-dimensional open subset of the exterior of $C$ which belong to
  other clusters. Because $\partial C$ starts to bend away from the
  flow line $\pi(\Real; x , f)$ at point $\xt$, it must cut through
  some flow lines that pass through the two open subsets, because the
  flow lines are smooth (three times differentiable) and, in a
  sufficiently small ball, their gradients can not deviate to any
  significant extent from the non-zero gradient of the flow line at
  the regular point $\xt$. Since no flow line can pass through areas
  of more than one cluster, this is a contradiction, which completes
  the proof.
\end{proof}

\begin{theorem}
  \label{thm:adjacent-clusters}
  There is at least a critical point in $M_{d - 1}$ on the border
  between two adjacent clusters, and the maximum of the density
  function over the border is attained at a critical point in
  $M_{d - 1}$.
\end{theorem}

\begin{proof}
  The border contains at least a critical point or a regular point in
  a stable $(d-1)$-manifold. In either case, the corresponding flow
  line stays on the boundaries of both clusters by
  Theorem~\ref{thm:boundary-flow-line}, and it thus belongs to the
  border between the two and so does its destination, a critical point
  in $M_{d-1}$ . The first claim thus follows.

  The second claim also follows easily, since every flow line on the
  border is either constant or strictly increasing, with its
  destination being a critical point in $M_{d - 1}$, also on the
  border.
\end{proof}

This result explains why clusters $C_3$ and $C_4$ in
Fig.~\ref{fig:bound} have no border, even though their boundaries are
extremely close to each other for $x_1$ large. This is simply because
there is no saddle point in this area.

A cluster is said to be \emph{connected} to another, if there is a
continuous path in $\Real^d$ connecting the two clusters, possibly
passing through any clusters and their common borders.

\begin{theorem}
  \label{thm:all-connected}
  All clusters are mutually connected.
\end{theorem}

\begin{proof}
  Consider an arbitrary proper sub-collection of clusters, say
  $\Ccal = \{C_i\}$, which are not necessarily mutually connected. Let
  $U = (\cup_i \Cbar_i)^\circ$. We only need to show that $U$ is
  connected to one cluster not in $U$. Note that $U$ is a proper
  subset of $\Real^d$ and $\partial U$ is a $(d-1)$-manifold. For any
  two adjacent $C_i, C_j \in \Ccal$, any $(d-1)$-manifold in $B_{ij}$
  must be in $U$. Only $\partial B_{ij}$ can possibly be in
  $\partial U$.  Since $\partial B_{ij}$ can only be a union of
  manifolds of dimensions lower than $d-1$ and there are only
  countably many $\partial B_{ij}$, we have
  $\partial U \backslash (\cup_{i,j} \partial B_{ij}) \ne
  \emptyset$. Choose any point in
  $\partial U \backslash (\cup_{i,j} \partial B_{ij})$. Then the flow
  line starting at it will remain on $\partial U$ and has a
  destination being a critical point $x^*$ of Morse index less than
  $d$, which is not in any $\bar B_{ij}$. The critical point $x^*$ has
  an unstable manifold, and the corresponding eigenvectors of the
  Hessian matrix of $f$ at $x^*$ are orthogonal to $\partial U$. Any
  flow line in the direction of leaving $U$ has a destination in the
  exterior of $U$. Hence $U$ is connected to at least one cluster
  outside, thus completing the proof.
\end{proof}

Theorem~\ref{thm:all-connected} is important for our hierarchical
clustering method, as it means that the clustering can proceed step by
step and will finish with one overall cluster, i.e., $\Real^d$.

\section{Cluster Significance and Merging}
\label{sec:method}

Modal clustering can proceed given a density function $f$, be it of
the population or an estimate from a sample. With the modes of $f$
identified, initial clusters can be determined according to the flow
lines as described in Section~\ref{sec:connected-clusters}. As
explained by \cite{hu-wang-2021}, one can be interested in reducing
the number of clusters, as $f$ may just be a density estimate
containing spurious modes, or as one may just want a simpler
clustering result for a fixed, smaller number of clusters even if
initial clusters are already statistically significant. By analogy,
one might focus on identifying individual hills and their distinct
peaks. However, it is also valuable to examine larger mountains—which
comprise multiple hills—or even entire mountain ranges. To achieve
this, \cite{hu-wang-2021} considered flatting a mode by increasing the
local bandwidth value. This redistributes the probability mass around
the mode to its surrounding region. However, this distorts the density
function increasingly as the bandwidth increases, and in certain cases
it may lead to the incorrect merging of major modal clusters before
merging the necessary minor ones.

To avoid density distortion, here we propose a truncation technique
that removes the bump associated with a cluster relative to all other
clusters.  This technique is somewhat related to the idea of ``excess
mass'' used by some multimodality tests
\citep{muller-sawitzki-1991,fisher-marron-2001,ameijeirasalonso-crujeiras-rodriguezcasal-2019},
where the excess mass is defined as the amount of mass that is to be
shifted to convert a bimodal distribution into a unimodal one.
However, our method does not fill the removed excess mass back in
anywhere else, which can itself be very tricky especially in a
high-dimensional space, and thus does not distort the density
elsewhere. This turns the density function $f$ into a sub-density
function, say, $\ft$, the integral of which is less than $1$. We then
use the Kullback-Leibler (KL) divergence \citep{kullback-leibler-1951}
of $\ft$ from $f$ to assess the (relative) significance of the
cluster:
\[
  \Delta(\ft, f) = \int_{\Real^d} f(x) \log
  \left[\frac{f(x)}{\ft(x)}\right] \ud x.
\]
This assessment is repeated for every cluster to find the least
significant one, i.e., the one with the smallest $\Delta$-value, which
is then merged into one of its neighbours that shares a critical point
in $M_{d-1}$. In the rare case when there are more than one such
neighbours, we simply choose the most significant one for cluster
merging. This may occur for a Morse function, because a Morse function
explicitly defines local properties around isolated critical points,
rather than prescribing global constraints on the flowlines. We prefer
to create a larger merged cluster, as it tends to give smaller
$\Delta$-values at subsequent merging steps.

The merging process starts with the initial clusters identified by the
modes of $f$ and merges one cluster into another at each step. Suppose
it is at the $k$-cluster level, i.e., when there are $k$ remaining
clusters, the collection of which is denoted by
$\Ccal^{(k)} = \{C_1^{(k)}, \ldots,C_k^{(k)} \}$. Assume that
$C_i^{(k)}$ is deemed the least significant and is to be merged into
$C_j^{(k)}$. Then the new cluster is the set
$C_i^{(k)} \cup C_j^{(k)}$---more precisely, the interior of the set
$C_i^{(k)} \cup C_j^{(k)} \cup B_{ij}^{(k)}$. With each merging of two
clusters, the number of clusters decreases by $1$. The merging
continues until all clusters coalesce into one, i.e., $\Real^d$,
through the process of which a hierarchical clustering tree is
produced.

The truncation technique is detailed as follows. Still assume the
merging process at the $k$-cluster level with sub-density function
$\ft^{(k)}(x)$. For cluster $C_i^{(k)}$, let
\begin{align}
  \label{eqn:fbar-ik}
  \fbar_i^{(k)} = \sup_{x \in \partial C_i^{(k)}} \ft^{(k)}(x),
\end{align}
i.e., the highest density value on its boundary. According to
Theorem~\ref{thm:adjacent-clusters}, $\fbar_i^{(k)}$ can only be
attained at a critical point of $f$ in $M_{d-1}$. The sub-density
function after truncating the bump of cluster $C_i^{(k)}$ at the
height of $\fbar_i^{(k)}$ is thus given by
\begin{align}
  \label{eqn:ft}
  \ft_i^{(k)}(x) = \left\{
    \begin{array}{ll}
      \ft^{(k)}(x), & \mathrm{~~if~~} x \notin C_i^{(k)}; \\
      \ft^{(k)}(x) \wedge  \fbar_i^{(k)}, & \mathrm{~~if~~} x \in C_i^{(k)},
    \end{array} \right.
\end{align}
i.e., only points in the $i$th cluster with density values greater
than $\fbar_i^{(k)}$ have their density values truncated to
$\fbar_i^{(k)}$.  The significance of a cluster is thus assessed with
$\Delta_i^{(k)} \equiv \Delta(\ft_i^{(k)}, f)$. Find $\Delta_i^{(k)}$
for all $i = 1, \dots, k$, and the least significant cluster is the
one with the smallest value of $\Delta_i^{(k)}$ and is to be merged
into one of its neighbours sharing the same critical point on their
border. Let cluster $C_{\bar i}^{(k)}$ be the least significant. Then,
with its merging away, the sub-density at the $(k-1)$-cluster level
becomes $\ft^{(k-1)} = \ft_{\bar i}^{(k)}$, which has the following KL
divergence from $f$:
\begin{align}
  \Delta(\ft^{(k-1)}, f)
  & = \int_{\Real^d} f(x) \log
    \left[\frac{f(x)}{\ft^{(k-1)}(x)}\right] \ud x \nonumber \\
  & = \int_{\Real^d} f(x) \log
    \left[\frac{f(x)} {\ft^{(k)}(x)} \cdot \frac{f^{(k)}(x)}
    {\ft^{(k-1)}(x)} \right] \ud x \nonumber \\ 
  & = \Delta(\ft^{(k)}, f)  + \int_{\Real^d} f(x) \log
    \left[\frac{\ft^{(k)}(x)}{\ft^{(k-1)}(x)}\right] \ud x,
  \label{eqn:KL-additive}
\end{align}
where $f = f^{(m)}$, $m$ being the number of initial clusters. Note
that $\ft^{(k)}$ and $\ft^{(k-1)}$ only differs in the region
corresponding to the truncated density bump of cluster
$C_{\bar i}^{(k)}$ and thus the integral in \eqref{eqn:KL-additive}
only needs to be evaluated in this region.

A rationale establishing that $\Delta_i^{(k)}$ reflects the
significance of the modal cluster $C_i^{(k)}$ is as follows. By
definition, $C_i^{(k)}$is classified as a modal cluster because it
contains a local mode. Truncating the density bump above the highest
density value on its boundary represents the threshold case required
to invalidate $C_i^{(k)}$ as a distinct modal cluster. This truncation
removes the minimum amount of probability mass necessary to eliminate
the mode. To assess the statistical significance of this lost
probability mass, the KL divergence serves as a natural and
mathematically sound choice. It provides a formal measure of how much
a sub-probability distribution differs from a reference distribution,
and it also maintains a direct connection to the likelihood framework
(see Section~\ref{sec:clustersignificance}).

\begin{figure}[!tbh]
  \centering
  \includegraphics[width=0.9\linewidth]{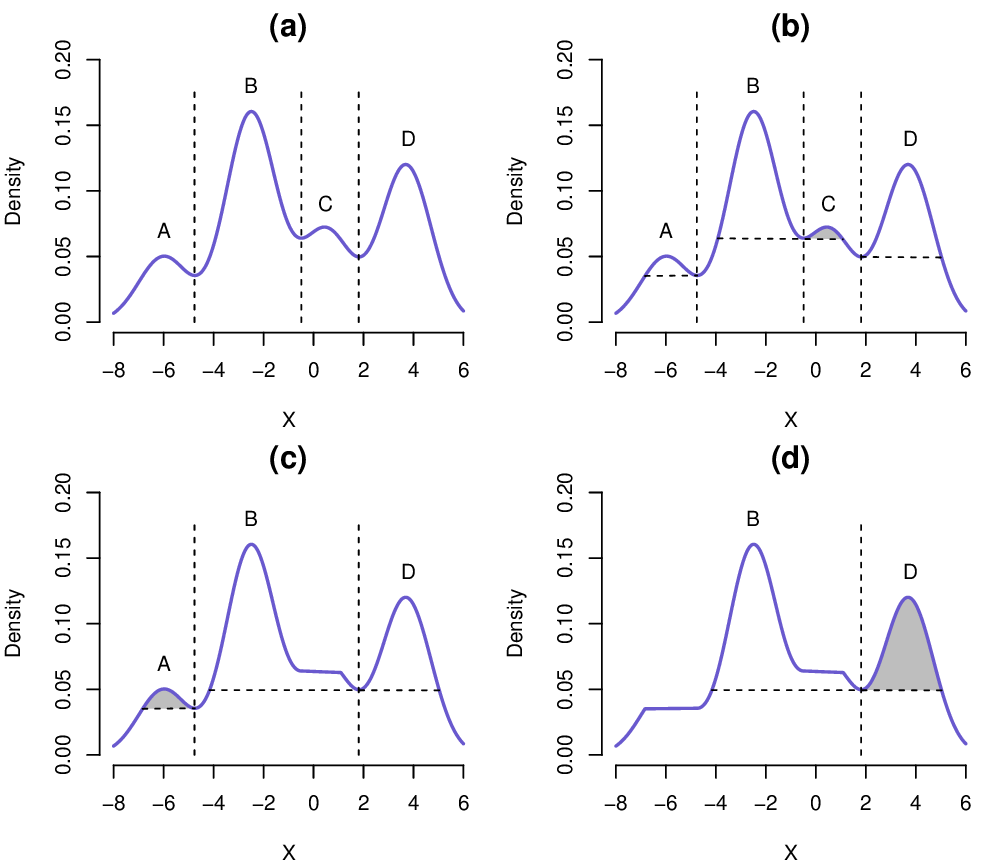}
  \caption{An illustrative example of the cluster-merging process. (a)
    Four initial modal clusters, separated by three local minima,
    i.e., the critical points in $M_{d-1}$; (b) The significance of
    each modal cluster is assessed by truncating its bump at the
    highest point on its boundary (a local minimum), and cluster $C$
    is found to be the least significant; (c) the bump of $C$ is
    truncated and $C$ is merged into $B$, and cluster $A$ now becomes
    the least significant of the remaining three and is to be merged
    into $B$; (d) $D$ is the less significant between the two and is
    to be merged into $B$ (actually $A \cup B \cup C$).}
  \label{fig:mmexample}
\end{figure}

Fig.~\ref{fig:mmexample} gives a simple, univariate example to
illustrate how the proposed method works. The density function $f$
shown in Fig.~\ref{fig:mmexample}(a) has four modal clusters $A$, $B$,
$C$ and $D$, as separated by three local minima shown with vertical
dashed lines. The significance of each modal cluster is then assessed
in terms of its $\Delta$-value if its bump is truncated at the level
of the highest density value on its boundary, as shown by each
horizontal dashed line in Fig.~\ref{fig:mmexample}(b).  The first
least significant cluster is found to be $C$. Truncating its shaded
area yields the smallest KL divergence compared to truncating each of
the other three clusters above their respective horizontal dashed
lines. After merging it into $B$, there remain three clusters, $A$,
$B \cup C$ and $D$, as shown in Fig.~\ref{fig:mmexample}(c), where the
original bump of $C$ has been truncated at the height of
assessment. The next least significant cluster is $A$. Truncating its
shaded area yields a smaller KL divergence than truncating the bump of
either remaining cluster ($B \cup C$ or $D$); consequently, cluster
$A$ is merged into $B \cup C$. Of the two remaining clusters, $D$ is
the less significant, as shown Fig.~\ref{fig:mmexample}(d), and
merging it into $A \cup B \cup C$ results in one final cluster,
covering the entire space of $\Real$. Note that at any stage of the
process, the significance of a cluster is assessed at the critical
point (of Morse index $d-1$) that has the highest density value on its
boundary. For each cluster, we can easily arrange these critical
points in the decreasing order of their density values and only
consider the very first one for potential truncation of the bump of
the cluster.

We name this cluster merging method the \textit{modal cluster merger}
(MCM). It is summarized step by step in Algorithm~\ref{alg:mcm}. We
can use a set of the indices of the modes to denote a cluster here,
and a clustering corresponds to a partition of the set
$\{1, \dots, m\}$, $m$ being the number of modes. The initial
clustering thus corresponds to the partition
$\Ical^{(m)} = \{I_1^{(m)}, \dots, I_m^{(m)}\}$, where
$I_i^{(m)} = \{i\}$ for $i = 1, \dots, m$. At the $k$-cluster level,
the partition is thus $\Ical^{(k)} = \{I_1^{(k)}, \dots, I_k^{(k)}\}$,
where each $I_i^{(k)}$ is a subset of $\{1, \dots, m\}$ and indicates
which initial modal clusters comprise the $i$th cluster at this
cluster level.

\begin{algorithm*}[!tbh]
  \caption{Modal Cluster Merger (MCM)}
  \label{alg:mcm}
  \begin{algorithmic}[1]
    \REQUIRE Density function $f$, with $m$ modes. 
    
    \STATE Set $k = m$
    
    \REPEAT  
    
    \STATE For each $I_i \in \Ical^{(k)}$, compute
    $\Delta_i^{(k)} = \Delta(\ft_i^{(k)}, f)$, evaluated by truncating
    its bump above the density height at its potential truncating
    point $x_i^{*(k)} \in M_{d-1}$.

    \STATE Find $I_{\bar i} \in \Ical^{(k)}$ such that
    $\bar{i} = \argmin_{i \in \{1, \dots, k\}} \Delta_i^{(k)}$.
    
    \STATE Find cluster $I_j \in \Ical^{(k)}$ (or the most significant
    cluster, if multiple) that is adjacent to cluster $I_{\bar i}$ at
    the merging critical point $x_{\bar i}^{*(k)}$.
    
    \STATE Replace $I_j$ with $I_j \cup I_{\bar i}$ and set 
    $\Ical^{(k-1)} = \Ical^{(k)} \backslash \{I_{\bar i}\}$

    \STATE Set $k = k - 1$
    
    \UNTIL {$k = 1$}
    
    \RETURN $\Ical^{(m)}, \ldots, \Ical^{(1)}$
  \end{algorithmic}
\end{algorithm*}

\section{Some Implementation Issues}
\label{sec:implementation}

This section details the implementation of Algorithm~\ref{alg:mcm},
with a specific focus on processing random samples. We outline and
discuss the key steps of this process, noting that alternative
methodologies and further optimizations are certainly
possible. Section~\ref{sec:density-estimation} describes a
semiparametric mixture approach to density
estimation. Section~\ref{sec:mode-identification} describes how to
identify the modes of the density estimate and at the same time
associate each observation with a
mode. Section~\ref{sec:critical-points} investigates how to compute
the critical points in $M_{d-1}$ of a density function. In
Section~\ref{sec:adjacency}, we present a practical, efficient
approach to identifying cluster adjacency. A fast way to estimate the
KL divergence between the original and a truncated density function is
described in Section~\ref{sec:clustersignificance}. Finally, time
complexity of the algorithm is analyzed in
Section~\ref{sec:complexity}.

\subsection{Density estimation}
\label{sec:density-estimation}

Density estimation is critically important for pursuing the ideal
population goal and a good density estimator is certainly helpful. To
implement the proposed clustering method, we use a semiparametric
mixture-based density estimator which was developed by
\citet{wang-wang-2015} and further modified by
\citet{hu-wang-2021}. The density estimator has been implemented for
solving clustering problems and has been shown to have fairly good
performance \citep{hu-wang-2021}. For self-containment purposes, we
briefly describe it as follows.

Given a dataset $\{x_i \in \mathbb{R}^d: i = 1, \ldots , n\}$, it
estimates a density function using a semiparametric Gaussian mixture:
\begin{align}
  \label{eqn:semiparametric-mixture}
  f(x; G, \Sigma) = \int \phi(x; \theta, \Sigma) \ud G(\theta),
\end{align} 
where $x$, $\theta$ $\in$ $\mathbb{R}^d$, $\phi (x; \theta, \Sigma)$
denotes the multivariate normal density with mean vector $\theta$ and
diagonal covariance matrix $\Sigma$. The mixing distribution function
$G$ is of a completely unspecified form by assumption, but it is known
that it always has a discrete maximum likelihood estimator and can
thus be exchanged with its support point set
$\Theta = (\theta_1, \dots, \theta_l)$ and the associated probability
mass vector $\omega = (\omega_1, \dots, \omega_l)$. Hence, density
\eqref{eqn:semiparametric-mixture} also has the following
homoscedastic finite mixture form:
\begin{align}
  \label{eqn:finite-mixture}
  f(x; G, \Sigma) = f(x; \Theta, \omega, \Sigma) =
  \sum_{j=1}^l \omega_j \phi(x; \theta_j, \Sigma),
\end{align}
though the number of components $m$ is unknown here and needs to be
determined via computation. The mixture is semiparametric, because $G$
is an infinite-dimensional parameter and $\Sigma$ is
finite-dimensional.

The log-likelihood function is given by
$l(G, \Sigma) = \sum_{i=1}^n \log f(x_i; G, \Sigma)$. However, its
direct maximization over both $G$ and $\Sigma$ is inappropriate, as it
always leads to degenerate $\Sigma$, with the likelihood value
approaching $\infty$ \citep{grenander-1981,geman-hwang-1982}. To
resolve this issue, one can consider the decomposition
$\Sigma = h^2 B$, where $B$ is symmetric, positive-definite and
subject to $|B| = 1$, with which the likelihood function
$\ell_h(G, B) = \ell(G, h^2 B)$ is bounded and can be maximized
accordingly, if $h > 0$ is held fixed.  The value of $h$ controls the
smoothness of the resulting density estimate
$f(\cdot; \Gh_h, h^2 \Bh_h)$, where
$(\Gh_h, \Bh_h) = \argmax \ell_h(G, B)$. An appropriate value of $h$
can then be determined using, e.g., a model selection criterion, such
as the corrected Akaike information criterion (AIC$_c$)
\citep{sugiura-1978,hurvich-tsai-1989}, which minimizes
\begin{align*}
  \mathrm{AIC}_c(h) = - 2 \ell_h(\Gh_h, \Bh_h) + 2 p + \frac {2 p (p +
  1)} {(n - p - 1)_+},
\end{align*}
where $p$ is the number of free parameter in a model and $a_+$ denotes
the positive part of $a$. For the semiparametric mixture, we use
$p = l (d + 1) + d - 1$, where $l$ is the number of mixture
components. In general, the AIC$_c$ has a fairly similar performance
to the AIC, unless $n$ is small or $p$ is large. With the
semiparametric mixture used for density estimation, $p$ can become
exceptionally large. Consequently, the AIC$_c$ serves as a critical
safeguard against extreme undersmoothing, preventing the selection of
an excessively small bandwidth ($h$-value).

For each fixed value of $h$, the likelihood can be maximized by a
combination of the constrained Newton method (CNM) \citep{wang-2007}
and the Expectation-Maximization (EM) algorithm
\citep{dempster-laird-rubin-1977}. The CNM is used to deal with the
nonparametric $G$ and the EM to update $\omega$, $\theta$ and $B$. We
also adopt a diagonal $\Sigma$ so that a multivariate normal component
can be decomposed into the product of univariate normals, thus giving
a computational advantage, especially in a higher-dimensional
situation. For more details, see \citet{wang-wang-2015} and
\citet{hu-wang-2021}.

Almost surely, the maximum likelihood estimator of the semiparametric
mixture density \eqref{eqn:finite-mixture}, with $h>0$ fixed, derived
from a random sample is a Morse function. This follows from the fact
that the log-likelihood function is continuously differentiable and
has a unique maximizer (for sufficiently many linearly independent
data points), combined with a parameter space that can be readily
compactified \citep{kiefer-wolfowitz-1956}. Because the parameter
estimates solve the corresponding likelihood equations, the Implicit
Function Theorem ensures that the maximum likelihood estimator is a
continuous function of the sample data. Consequently, the probability
that a density estimate possesses any critical point with a degenerate
Hessian matrix is zero.

The maximum likelihood estimator of the semiparametric mixture density
also has a finite number of components and thus a finite number of
modes \citep{lindsay-1983,alexandrovich-holzmann-ray-2013}.
Consequently, Assumption (A) holds almost surely for our density
estimator given a random sample. While the proposed clustering method
is suitable for any problem where density estimation is appropriate,
tuning $h$ further provides a dense sieve of mixture densities capable
of approximating any arbitrary target density.

\subsection{Mode identification and observation association}
\label{sec:mode-identification}

To use the MCM, we first need to identify all critical points in
$M_d$, i.e., the modes. By definition, a point is associated with a
mode by climbing up its flow line. With numerical data, this means one
should use the steepest ascent method, with an infinitesimal step
size, but this is impractical. One can practically use a larger step
size, but it is not so easy for the user to choose an appropriate step
size. Choosing it too large, it ends up in a wrong modal cluster
and choosing it too small, it takes a long time to reach a mode. 

Therefore, for this purpose, we resort to the Modal EM (MEM) algorithm
\citep{li-ray-lindsay-2007}, possibly with a partial step size. Since
each observation also needs to be associated with a mode, we can start
the algorithm with each observation, which will produce both the modes
and the modal clusters of the observations.

The algorithm is fairly straightforward for a Gaussian mixture. For
the homoscedastic mixture \eqref{eqn:finite-mixture}, the iterative
formula for updating $x$ to $x'$ is
\begin{align}
  \label{eqn:mem}
  x' = (1 - \alpha) x + \alpha \sum_{j=1}^l p_j \theta_j,
\end{align}
for a step size $0 < \alpha \leq 1$, where, denoting
$f(x) = f(x; G, \Sigma)$,
\[
  p_j = \frac {\omega_j \phi(x; \theta_j, \Sigma)} {f(x)}, ~~ j = 1,
  \dots, l.
\]
Owing to its minorization property \citep{lange-2013}, the MEM
algorithm is particularly effective for associating an observation
with a nearby mode. For points far from the mode, a partial step size
may be necessary, especially during the initial iterations. While the
ascent path generated by the MEM algorithm differs slightly from that
of the steepest ascent method, the resulting modal clusters are
practically identical. We provide a brief discussion of this behavior
in Section 7. A major advantage of the MEM algorithm over steepest
ascent is its automatic step size ($\alpha = 1$), which performs
exceptionally well provided the starting point is not too far away
from the mode.

The set of unique modes for the initial density estimate is then
determined by the destinations reached from all observations via the
MEM. Each observation is then also associated with the mode it
reaches. If a modal cluster is to be merged into another one, then all
observations associated with the first are to be merged into the
second.

In rare cases, there may be modes not reached by any observation. This
causes no practical issue for the MCM.  Since there is no observation
associated, these modes are minor. If one finds such modes in any way,
they will be the first ones absorbed into their neighbouring clusters
by the practical cluster merging method described in
Section~\ref{sec:clustersignificance}.

\subsection{Finding critical points in $M_{d-1}$}
\label{sec:critical-points}

To find a non-mode critical point, we consider using the Newton-Raphson
(NR) method that iteratively updates $x$ with the following formula:
\begin{align}
  \label{enq:nr-formula}
  x' = x - \alpha [\nabla^2 f(x)]^{-1} \nabla f(x),
\end{align}
where $ 0 < \alpha \leq 1$ is a step size, and $\nabla$ and $\nabla^2$
denote, respectively, the gradient and Hessian operators with respect
to $x$. For a density $f$, it is often computationally more efficient
using $\log f$ in \eqref{enq:nr-formula}, instead of $f$, and it is
also numerically more stable. Note that $\log f$ and $f$ have the same
critical points and their Morse indices.

The NR method with a full step size ($\alpha = 1$) converges
automatically in a sufficiently small neighbourhood of a critical
point with a non-singular Hessian, a property satisfied by a Morse
function. The NR method has a quadratic order of convergence and often
finds a nearby critical point in a few iterations. To find a critical
point in $M_{d-1}$, we modify the NR method and adopt a Monte Carlo
approach to starting it. The NR method starts with a random point and
always examines $\lambda(x)$, the number of the negative eigenvalues
of the Hessian at point $x$. The computation is aborted in its first
iteration, if $\lambda(x) \ne d-1$. Otherwise, it continues and moves
forward inside the region with $\lambda(x) = d-1$. If an iterate steps
out of this region, the step size is halved repeatedly until the
property $\lambda(x) = d-1$ is satisfied. The iteration continues
until a critical point is found, or is aborted if a pre-fixed maximum
number of iterations (10, say) is reached. The iteration can also be
aborted, should any numerical difficulty arises.

To generate random starting points, we initially considered using the
estimated mixture density, possibly with a scaled-up $\Sigma$. While
it seems nice that these points are distributed in the area where the
critical points are, it did not work well. The problem is that it has
a high probability to locate critical points in high density regions
but a low probability to locate those in low density regions. As a
result, the NR method needs to run a large number of times to find all
critical points in $M_{d-1}$ when there are critical points in low
density regions. Because each critical point requires only a single
visit and the objective is to minimize total collection time, the
optimal strategy assumes a uniform spatial distribution. Consequently,
we employ a uniform distribution on a support containing the critical
points. In a low-dimensional space, it seems convenient enough to use
a cube as the support of the uniform distribution, which can be
determined easily from the range of data or the mixture component
centers. In a high-dimensional space, the convex hull of the mixture
component centers, possibly expanded slightly from its gravity center,
is used as the support. Note that all critical points of a
homoscedastic mixture density must be inside this convex hull
\citep{ray-lindsay-2005}. This strategy works quite well in all of our
numerical studies, although further improvements are highly
possible. In addition, if one is more worried about the risk of not
finding those critical points in high density areas, it is then
possible to use the mixture density estimate to generate a portion of
the starting values.

The repeated runs of the NR method can be terminated once all critical
points have presumably been identified. To ensure completeness, one
must verify that all modes are mutually connected via these discovered
critical points. Furthermore, a sufficient number of additional runs
should be performed after locating each critical point, and a minimum
threshold of total repetitions should better be enforced in every
scenario.

As an illustration, we ran our R implementation of this strategy for a
fixed 1000 repetitions of the NR method, for the Quadrimodal
distribution shown in Fig.~\ref{fig:bound}, using the uniform
distribution on $[-2,2]^2$ to generate random starting points. The
resulting frequencies for the four saddle points (numbered in
ascending order of $x_1$) and the running time are listed in the left
half of Table~\ref{tab:freq}. Since each NR run is started with an
independent random point, it should take at most dozens of runs to
have a full collection of the saddle points for such a
distribution. The probability of locating a specific saddle point
equals the ratio of its basin of attraction's measure (area or volume)
to the total support of the uniform distribution.

\begin{table}[!tbh]
  \centering
  \begin{tabular}{ccccccc c ccccccc} \hline
    \multicolumn{7}{c}{Saddle Points} && \multicolumn{7}{c}{Modes} \\
    \cmidrule{1-7} \cmidrule{9-15}
    S1 & S2 & S3 & S4 & Aborted && Time (s) && M1 & M2 & M3 & M4 & Aborted && Time (s) \\ 
    64 & 81 & 122 & 69 & 664 && 1.28 && 211 & 128 & 113 & 200 & 348 && 1.16 \\ \hline
  \end{tabular}
  \caption{Frequencies of the saddle points and modes found by the
    modified NR method for the Quadrimodal distribution}
  \label{tab:freq}
\end{table}

The same strategy can also be used to find the modes of the density
function, by ensuring $\lambda(x) = d$, or in fact critical points of
any Morse index. Finding the modes this way is faster if the data
dimension is not very high, and it only involves the estimated
density, not the observations at all. Once the modes are identified,
observation association using the MEM becomes more straightforward, as
an association can be established as soon as the MEM iterate
approaches a mode, without needing to reach it. The resulting
frequencies for the four modes out of 1000 NR runs are listed in the
right half of Table~\ref{tab:freq}. Certainly, to find both saddle
points and modes together, an NR run should not be aborted unless the
collection of the critical points of the corresponding Morse index is
deemed complete.

\subsection{Establishing cluster adjacency}
\label{sec:adjacency}

Once the necessary critical points are identified, we need to
establish pairwise cluster adjacency, by identifying all clusters that
share a critical point in $M_{d-1}$ on their boundaries. This can be
achieved by generating points in a sufficiently small ball centred at
the critical point and identifying the modes that are the destinations
of the flow (or MEM) lines of these points. Points can be generated
either systematically or randomly; however, we adopt a random
generation approach in our implementation. To maximize the success
rate of locating distinct modes, we consistently employ pairs of
reflection points symmetric across the critical point.

Since no found critical point can be numerically perfect, using too
small a ball centered at a critical point for random point generation
can result in the generated points all climbing up to an identical
mode, and using too large a ball may cause them going to wrong modes
at a distance. To resolve this issue, we start with a sufficiently
small ball and then gradually double the radius of the ball, until for
the first time the generated points reach distinct modes.

After identifying the adjacency among the clusters, the merging
process can proceed easily. One only needs to compute the density
values at the critical points in $M_{d-1}$. For each cluster, its
density truncation can only occur in the top-down order in terms of
the height of the density at these critical points on its boundary.

\subsection{Assessing cluster significance}
\label{sec:clustersignificance}

As stated in Section~\ref{sec:method}, the KL divergence is an
efficient way to measure the difference between two (sub-)density
functions. To avoid numerical integration over an irregular region, we
can conveniently find their approximate values from the given data.

Consider a consistent density estimator $\fh$ based on a random sample
$X = \{x_1, \dots, x_n\}$. The KL divergence
$\Delta(\ft_i^{(k)}, \fh)$ can be well approximated by
\begin{align}
  \widehat{\Delta}(\ft_i^{(k)}, \fh)
  & = \frac 1 n \sum_{x \in X} \log \left[\frac {\fh(x)} {\ft_i^{(k)}(x)}
    \right] \nonumber \\
  & =  \widehat{\Delta}(\ft^{(k)}, \fh)
    + \frac 1 n \sum_{x \in C_i^{(k)}}\log\left[\frac {\ft^{(k)}(x)}
    {\ft_i^{(k)}(x)}\right],
    \label{eqn:delta-sample}
\end{align}
where $\ft_i^{(k)}$ is as given in (\ref{eqn:ft}), with $f$ replaced
by $\fh$. This empirical version of the KL divergence,
$\widehat{\Delta}(\ft_i^{(k)}, \fh)$, has a sound statistical basis
for assessing the significance of cluster $C_i^{(k)}$, as
$n \widehat{\Delta}(\ft_i^{(k)}, \fh)$ is exactly the log-likelihood
ratio between $\fh$ and $\ft_i^{(k)}$. Note that the summation in
\eqref{eqn:delta-sample} only needs to be evaluated over the data
points in $C_i^{(k)}$ that have their density values truncated.

Once cluster $C_{\bar i}^{(k)}$ is chosen and merged away, set
$\ft^{(k-1)} = \ft_{\bar i}^{(k)}$. The empirical KL divergence
$\widehat{\Delta}(\ft^{(k)}, \fh)$ increases as $k$ decreases, and
$n \widehat{\Delta}(\ft^{(k)}, \fh)$ can be interpreted as the
log-likelihood loss of $\ft^{(k)}$ relative to $\fh$.

\subsection{Time complexity}
\label{sec:complexity}

The computational cost of the algorithm lies dominantly on density
estimation and critical point finding. Density evaluation at all data
points requires $O(ndl)$ basic operations, where $n$ denotes the
sample size, $d$ the dimension of the data and $l$ the number of
mixture components required. This means that mode identification and
observation association, which uses the MEM algorithm, requires
$O(ndl)$ operations per iteration. For density estimation, the CNM
further requires $O(nl^2)$ operations for updating the mixing
proportions $\omega_j$'s in each iteration.  It almost always
terminates within dozens of iterations.

For the modified NR method used for finding critical points in
$M_{d-1}$, inverting a Hessian matrix requires $O(d^3)$ operations. To
find all $k$ critical points in $M_{d-1}$, it thus requires $O(kd^3)$
in probability, if each critical point has a positive probability of
being located, and the number of repeated runs may be up to a few
hundreds, depending on the size of $M_{d-1}$. This is a daunting task
if $d$ is large.

Numerous algorithms exist in the literature that specifically target
non-mode critical points; see, e.g., \cite{conn-gould-toint-1991},
\cite{weinan-zhou-2011}, \cite{asl-lu-yang-2022} and
\cite{liu-luo-2022}. These algorithms require only $O(d^2)$ operations
per iteration. We will evaluate their applicability to our current
problem in future studies.

The numerical performance of the current implementation of MCM is
investigated and discussed in Section~\ref{sec:numeric}, with running
times given in Section~\ref{sec:conclusion}.

\section{Numerical Studies}
\label{sec:numeric}

\subsection{Setup}

In this section, we study the performance of the proposed clustering
method and compare it with several other clustering methods, including
K-means, Spectral clustering (SpecClust), kernel K-means (KK-means),
DBSCAN, PdfCluster, HMAC and MDE-MF. The inclusion of K-means here is
due to its simplicity and popularity as a clustering method, although
it is not a mode-based but a distance-based one. SpecClust and
KK-means are two kernel-based methods and can possibly find non-convex
clusters. DBSCAN is a well-known density-based clustering algorithm
and excels at discovering complex, irregular shapes while filtering
out background noise. PdfCluster, HMAC and MDE-MF all aim to find
modal clusters directly.

All computations have been carried out in R \citep{citeR}. For
K-means, SpecClust and KK-means, DBSCAN, PdfCluster and HMAC, the
following R functions are used, respectively: \texttt{stats::kmeans},
\texttt{kernlab::specc}, \texttt{kernlab::kkmeans},
\texttt{dbscan::dbscan}, \texttt{pdfCluster::pdfCluster} and
\texttt{Modalclust::phmac}
\citep{karatzoglou-smola-hornik-2024,hahsler-piekenbrock-2025,azzalini-menardi-2014,cheng-ray-2014}. For
all problems studied below, we largely focus on the results with a
pre-fixed number of clusters. If a function returns a hierarchical
tree, we use the major clusters with the desired number of clusters by
cutting the tree. For \texttt{dbscan}, we find the clusters of the
desired number by controlling the radius of the epsilon neighborhood
(\texttt{eps}). The two practical datasets studied in
Sections~\ref{sec:simul4} and \ref{sec:simul5} are standardized so
that each variable has mean $0$ and variance $1$, as often preferred by
a clustering algorithm in practice. We note that the performance of
MCM and MDE-MF is invariant of such data transformation, owing to the
mixture-based density estimator they use.

Wherever possible, the adjusted Rand index (ARI)
\citep{rand-1971,hubert-arabie-1985} is used to measure the similarity
between the partitioning result of a clustering method and the true
partition (or class labels if available). The ARI value ranges between
$-1$ and $1$. A larger positive ARI value indicates a higher level of
similarity between the two partitions, while a value of $0$ is
equivalent to completely random allocations and a negative value
indicates agreement worse than random allocations.

We prefer to use the default settings of these R functions to avoid
unnecessary complications. However, we notice that \texttt{dbscan} has
an argument \texttt{minPts} (defaulted to $5$, for the number of
minimum points required in the epsilon neighborhood for core points)
and the value it takes can substantially affect the clustering
result. We hence set it to a value in each case so as to achieve the
optimal or near-optimal clustering outcome, in terms of the resulting
ARI values or visual inspection. The DBSCAN method is thus denoted
along with the tuned value of \texttt{minPts}, as
DBSCAN(\texttt{minPts}). We notice that \texttt{dbscan} tends to crash
if a larger value for \texttt{minPts} ($> 20$, say) is used. It is
often quite a challenge in practice to find suitable values for such
hyperparameters of a clustering algorithm, as rarely there is an
independent, objective criterion for evaluating a clustering method,
and it is also not easy to visualize the clustering results even in a
moderate-dimensional space. By default, \texttt{specc} and
\texttt{kkmeans} automatically find a value for the width parameter
(\texttt{sigma}) of the radial basis function kernel. This seems to
generally work fine, although using some fixed values for
\texttt{sigma} may result in higher ARI values in some cases. Both
\texttt{specc} and \texttt{kkmeans} are a little unstable and, if
repeated, may not produce the same clustering results for the same
data. We hence report their average performance wherever possible.

\subsection{Data from a quadrimodal distribution}
\label{sec:simul1}

The first case study is a simulation using the Quadrimodal
distribution shown in Fig.~\ref{fig:bound}; see \cite{wand-jones-1993}
for the parameters of this distribution. This seems to be a fairly
typical clustering scenario. In each simulation run, $1000$
observations are randomly generated from the distribution.  The true
partition of the generated observations is obtained by applying the
MEM algorithm using the true density. Each clustering method is then
applied to the generated data and the resulting partition at the
$4$-cluster level is compared against the true partition, as the true
density has $4$ modes.

Using a generic random sample of this distribution,
Fig.~\ref{fig:simu1_result} shows the true partition and the
partitions produced by the 8 clustering methods included in our
study. It can be seen that the allocations of the observations by MCM
and MDE-MF are the most similar to the true partition, followed by
K-means and SpecClust. KK-means and PdfCluster find 4 clusters of
reasonable sizes but they dot not match well the true partition. Due
to kernel transformation, SpecClust and KK-means can group widely
separated data points into the same cluster, which may produce
counterintuitive results. HMAC has found 3 relatively large clusters,
but the fourth one contains only 1 outlying data point. Despite the
best choice for \texttt{minPts}, DBSCAN still fails to effectively
isolate modal clusters when data points are dense in between. It also
labels a number of data points as noise (light grey points).

\begin{figure}[!tbh]
  \centering
  \includegraphics[width=0.9\linewidth]{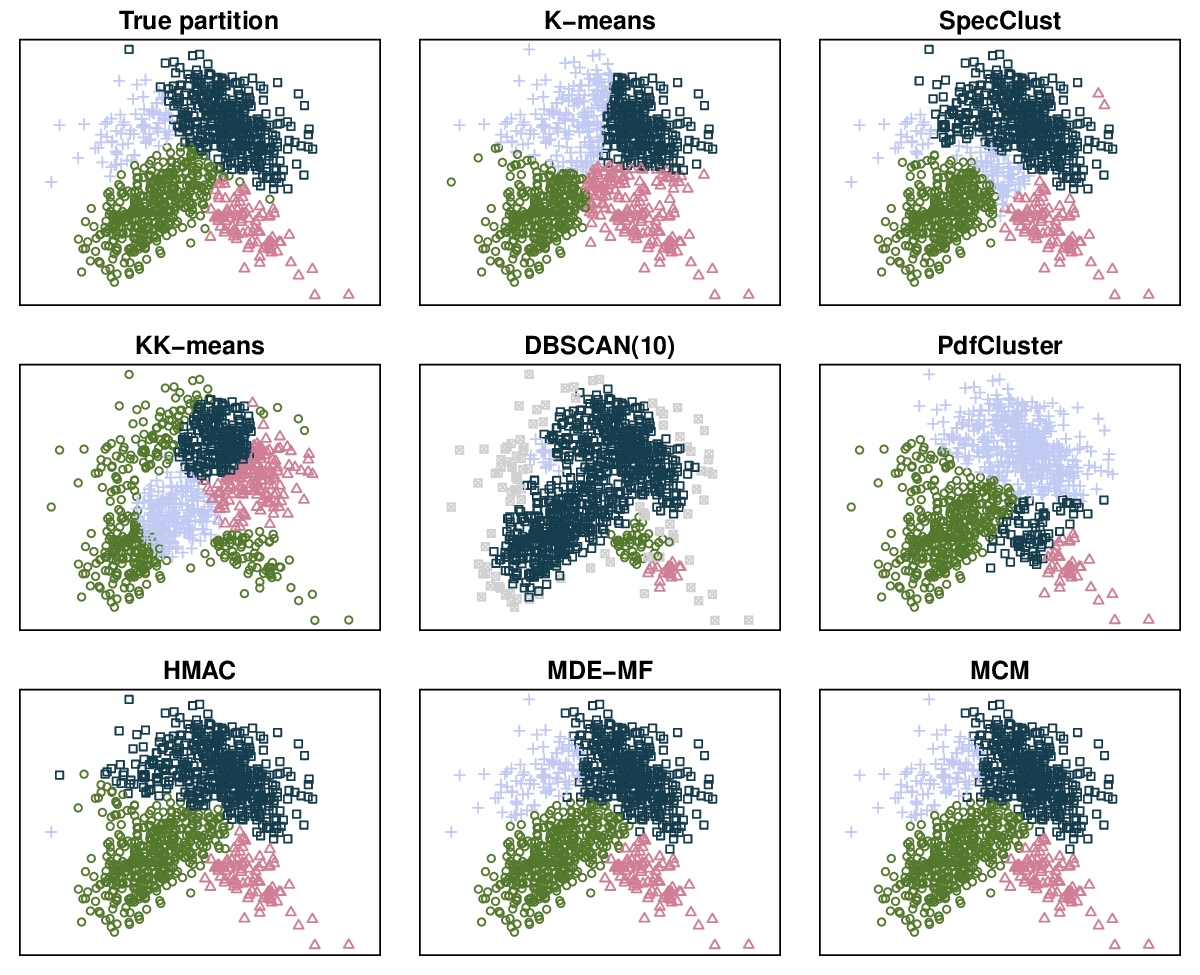}
  \caption{Results at the $4$-cluster level in the scenario presented in 
    Section~\ref{sec:simul1}}
  \label{fig:simu1_result}
\end{figure}

The simulation is repeated $50$ times and the average ARI values of
the clustering methods are given in Table~\ref{tab:simul1rand}.  The
proposed MCM has the highest ARI value (shown in bold, also in
Tables~\ref{tab:simul2rand}--\ref{tab:simul5result}), which is closely
followed by MDE-MF. All other methods have apparently smaller ARI
values, with DBSCAN giving the lowest.

\begin{table}[!tbh]
  \centering
  \begin{tabular}{cccccccc} \hline		
    K-means & SpecClust & KK-means & DBSCAN(10) & PdfCluster & HMAC & MDE-MF & MCM \\
    \hline
    $0.659_{0.005}$ & $0.701_{0.010}$ & $0.471_{0.012}$ & $0.301_{0.032}$ & $0.796_{0.021}$ & $0.734_{0.036}$ & $0.927_{0.005}$ & $\mathbf{0.931}_{0.006}$
    \\ \hline
  \end{tabular}
  \caption{Mean ARI values at the $4$-cluster level in the scenario
    presented in Section~\ref{sec:simul1}.  Standard errors are given
    in subscripts. Default argument settings are used for SpecClust
    and KK-means. Out of all feasible values, \texttt{minPts=10} gives
    the largest mean ARI values for DBSCAN.}
  \label{tab:simul1rand}
\end{table}

\subsection{Data structured in three concentric circles}
\label{sec:simul2}

The second scenario has a structure of $3$ concentric circles, similar
to the example given by \citet{hastie-tibshirani-friedman-2009}, page
546, where they used the special technique of spectral clustering to
solve this problem. Here, we generate the observations of the three
groups uniformly distributed on three circles with radii of $1$, $2.8$
and $5$ respectively, and then add to each observation a bivariate
Gaussian fluctuation with mean $0$ and standard deviation of $0.25$ in
any direction. The numbers of observations on the inner, middle and
outer circles are $150$, $200$ and $250$, respectively. Ideally, a
modal clustering algorithm should produce three circular clusters,
corresponding to the three uniform distributions on the circles.

Fig.~\ref{fig:simu2_bound}(a) shows a generic dataset (circles) and a
contour plot of the density estimate produced by the algorithm
described in Section~\ref{sec:density-estimation}. The estimated
density has 39 modes (solid points) and 62 saddle points (triangles),
which are all connected in one network (line segments) according to
modal cluster adjacency established by the methods described in
Sections~\ref{sec:mode-identification} and \ref{sec:adjacency}.
Fig.~\ref{fig:simu2_bound}(b) presents the 39 initial modal clusters,
arbitrarily labeled at their modes. Adjacent clusters are separated by
solid curves, with one saddle point residing on a border.  It is
interesting to note that only a minimum (not shown) can be bordering
more than one cluster.

\begin{figure}[!tbh]  
  \centering
  \includegraphics[width=0.95\linewidth]{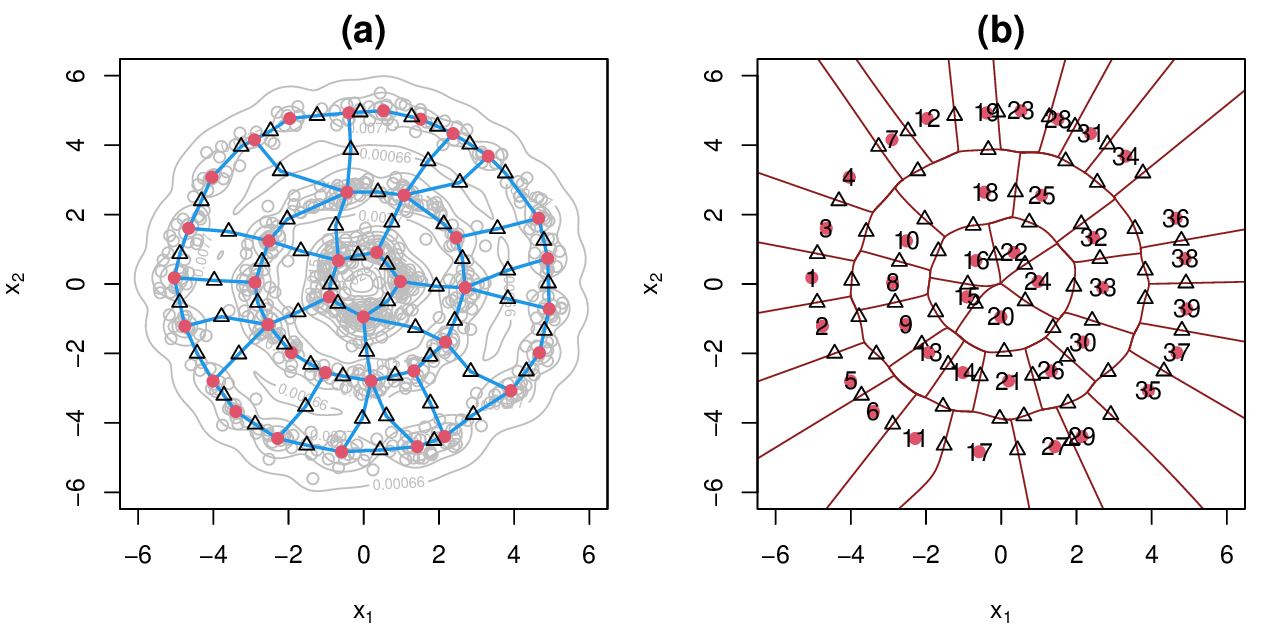}
  \caption{(a) The contour plot of the density estimate produced by
    MCM in the scenario presented in Section~\ref{sec:simul2},
    superimposing 600 generic data points (circles), along with the
    established network connecting all modes (solid points) and saddle
    points (triangles); (b) Initial modal clusters, arbitrarily
    labelled, in the scenario presented in Section~\ref{sec:simul2}.}
  \label{fig:simu2_bound}
\end{figure}

This is a very challenging problem for modal clustering, as an
estimated density can easily contain many minor modes along each
circle. The nonparametric density estimate here looks pretty good
already, as shown in Fig.~\ref{fig:simu2_bound}(a). Gradually
increasing the bandwidth---either globally or locally---progressively
distorts the density estimate. This distortion can raise the valleys
between the three large circular clusters too high at some location
before all minor clusters along a circle merge into one.

Fig.~\ref{fig:simu2_others} shows the results of the competing
clustering methods at the $3$-cluster level on the generic
dataset. Not surprisingly, $k$-means is not able to deal with this
scenario because the true centres of the three clusters are
identical. PdfCluster and HMAC also perform poorly as increasing the
global bandwidth value severely distorts the original density
estimate. MDE-MF flattens the density estimate locally and it does not
perform well either in this case, because it merges minor clusters on
different circles too early. In aid of kernel transformation, KK-means
still does not perform well. If the tuning parameters are carefully
chosen for the generic dataset here with the clustering goal in mind,
SpecClust and DBSCAN can perfectly produce the three circular
clusters.

\begin{figure}[!tbh]
  \centering
  \includegraphics[width=0.95\linewidth]{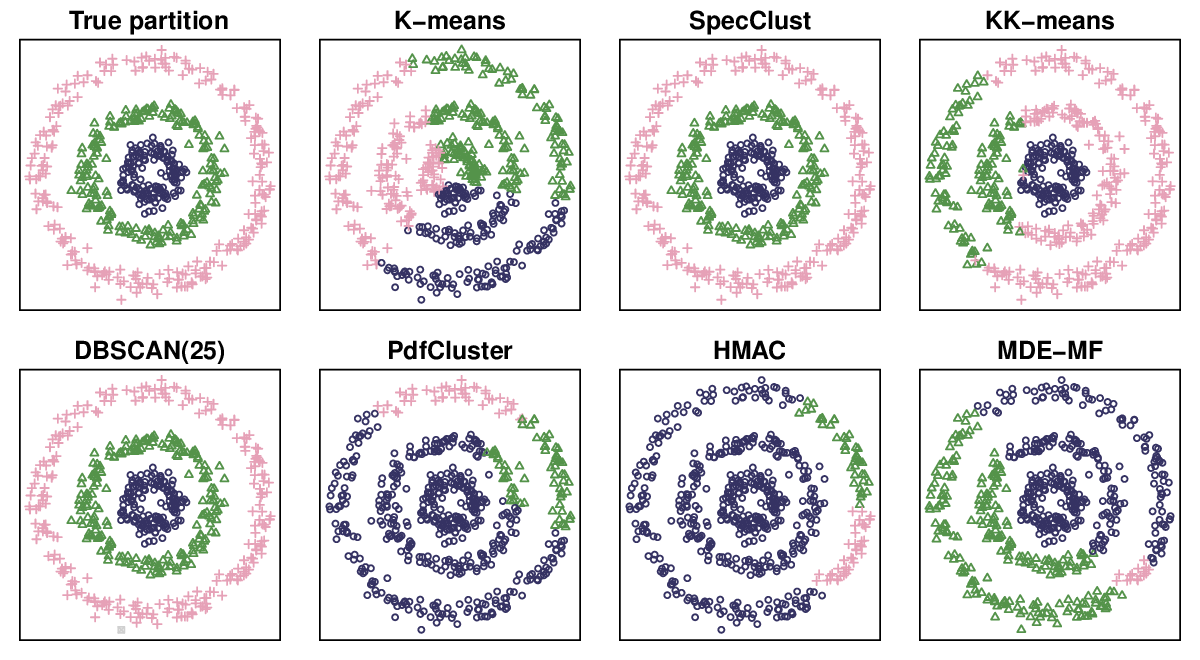}
  \caption{Results of the competing clustering methods at the
    $3$-cluster level in the scenario presented in
    Section~\ref{sec:simul2}. Tuning parameter values have been
    selected for best visual effects.}
  \label{fig:simu2_others}
\end{figure}

\begin{figure}[!tbh]  
  \centering
  \includegraphics[width=0.95\linewidth]{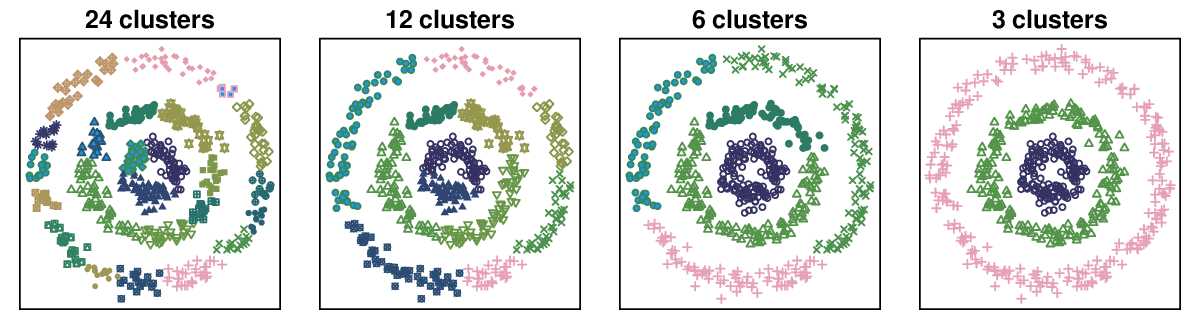}
  \caption{Results of MCM at several clustering levels in the scenario
    presented in Section~\ref{sec:simul2}}
  \label{fig:simu2_our}
\end{figure}

In comparison, Fig.~\ref{fig:simu2_our} shows the results of MCM at
the $24$-, $12$-, $6$- and $3$-cluster levels. These results at all
cluster levels are intuitively appealing---small or big gaps can be
seen between clusters and clusters with smaller gaps are merged
earlier. At the $6$-cluster level, the initial clusters on the inner
circle forms a single cluster, and the initial clusters on the middle
and outer circles are merged into several larger clusters. Finally, at
the $3$-cluster level, the observations are grouped into three
circular clusters, as desired. The reason that MCM produces such a
perfect partition is that it only truncates the minor bumps inside
each circular cluster and does not at all distort the density in the
valleys between the three large circular clusters.

The dendrogram given in Fig.~\ref{fig:simu2_dendrogram} illustrates
the complete cluster-merging process, where the branches and nodes
along the tree show how the clusters are merged in steps and the
values of the $y$-axis represent
$n \widehat{\Delta}(\ft_i^{(k)}, \fh)$, the log-likelihood loss
relative to the initial density estimate. We can see a drastic loss in
log-likelihood when the number of clusters reduces from $3$ to $2$,
suggesting that the true number of clusters is $3$.

\begin{figure}[!tbh]  
  \centering
  \includegraphics[width=0.9\linewidth]{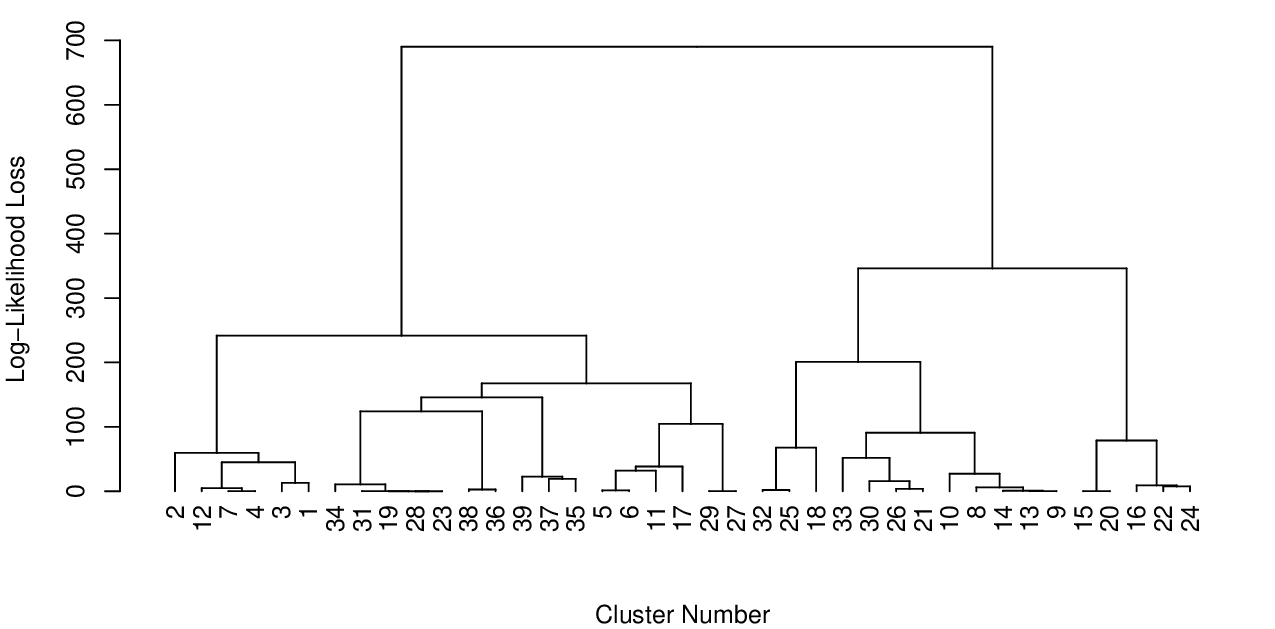}
  \caption{Dendrogram of the log-likelihood loss at each
    cluster-merging step in the scenario presented in
    Section~\ref{sec:simul2}}
  \label{fig:simu2_dendrogram}
\end{figure}

The ARI values averaged over $50$ simulation runs are given in
Table~\ref{tab:simul2rand}. Under repetitions, only MCM can reliably
partition observations into three perfect circular clusters.

\begin{table}[!tbh]
  \centering
  \begin{tabular}{cccccccc} \hline
    K-means & SpecClust & KK-means & DBSCAN(3) & PdfCluster & HMAC & MDE-MF & MCM \\ \hline
    $0.003_{0.001}$ & $0.700_{0.035}$ & $0.333_{0.016}$ & $0.723_{0.033}$ & $-0.008_{0.005}$ & $0.043_{0.019}$ & $-0.003_{0.005}$ & $\mathbf{1.000}_{0.000}$ \\ \hline
  \end{tabular}
  \caption{Mean ARI values at the 3-cluster level in the scenario
    presented in Section~\ref{sec:simul2}. Standard errors are given
    in subscripts. Default argument settings are used for SpecClust
    and KK-means. Out of all feasible values, \texttt{minPts=3} gives
    the largest mean ARI value for DBSCAN.}
  \label{tab:simul2rand}
\end{table}

\subsection{T4.8k data}
\label{sec:simul3}


In this section, we study the famous, artificial T4.8k data
\citep{karypis-han-kumar-1999}.  It contains $8,000$ observations with
two numeric variables and is shown in
Fig.~\ref{fig:simu3_result}(a). Visually, the dataset consists of six
main irregularly shaped clusters alongside scattered noise points, as
well as an additional sinusoidal structure that is virtually
impossible to detect as a distinct cluster.

\begin{figure}[!tbh]  
  \centering
  \includegraphics[width=0.95\linewidth]{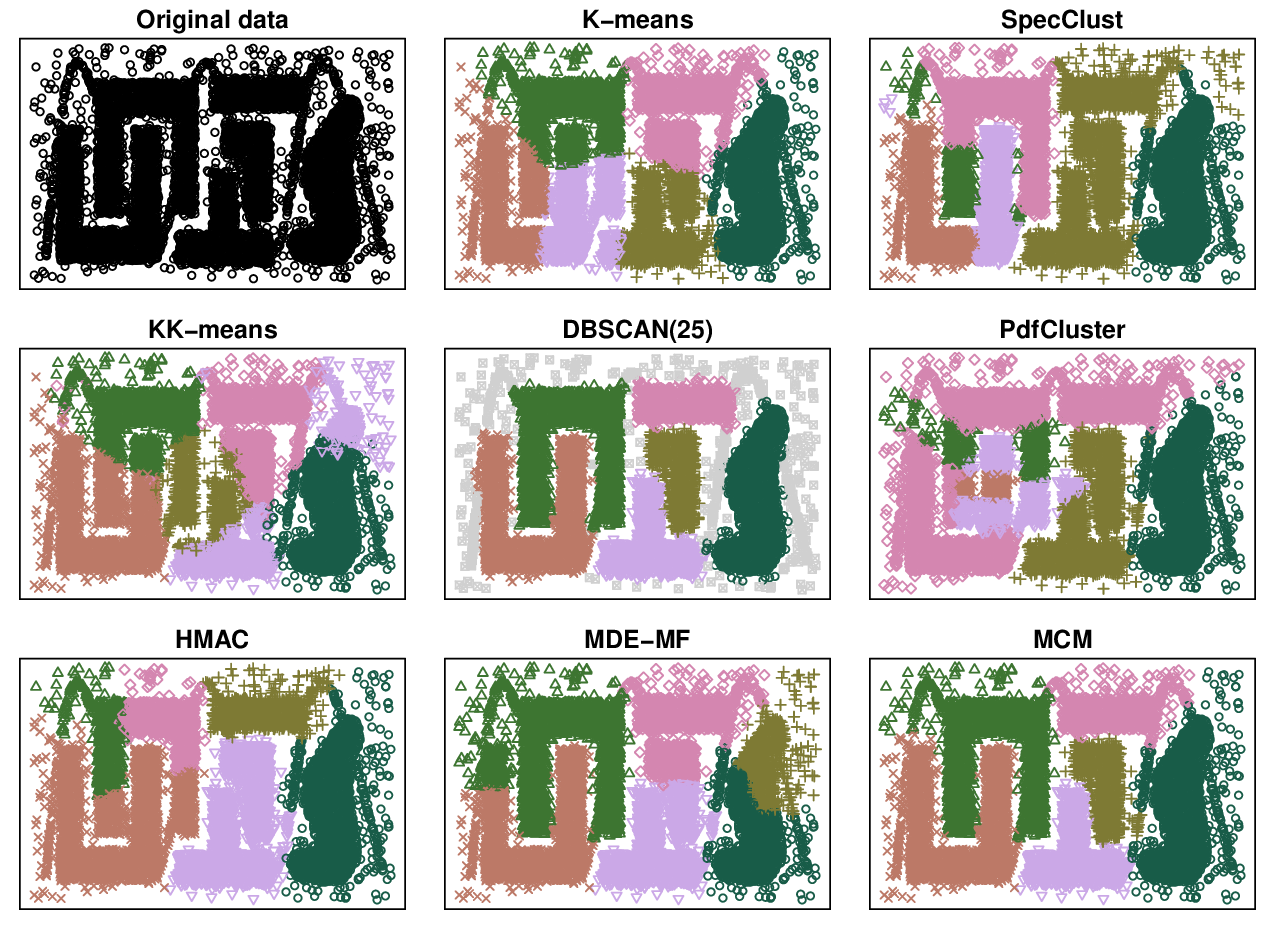}
  \caption{Results at the $6$-cluster level for T4.8k data. Tuning
    parameter values have been selected for best visual effects.}
  \label{fig:simu3_result}
\end{figure}

The results of the clustering methods at the $6$-cluster level are
given in Fig.~\ref{fig:simu3_result}. We can see that MCM successfully
identified the six major clusters, aligning well with visual
inspection. DBSCAN, with \texttt{minPts} optimally set to 25, also
produced the 6 main clusters, while labeling some edgy points as noise
(light grey points). SpecClust, MDE-MF and HMAC have identified some
main clusters of irregular shapes, while K-means, KK-means and
PdfCluster largely failed in this scenario.

\begin{figure}[!tbh]  
  \centering
  \includegraphics[width=1\linewidth]{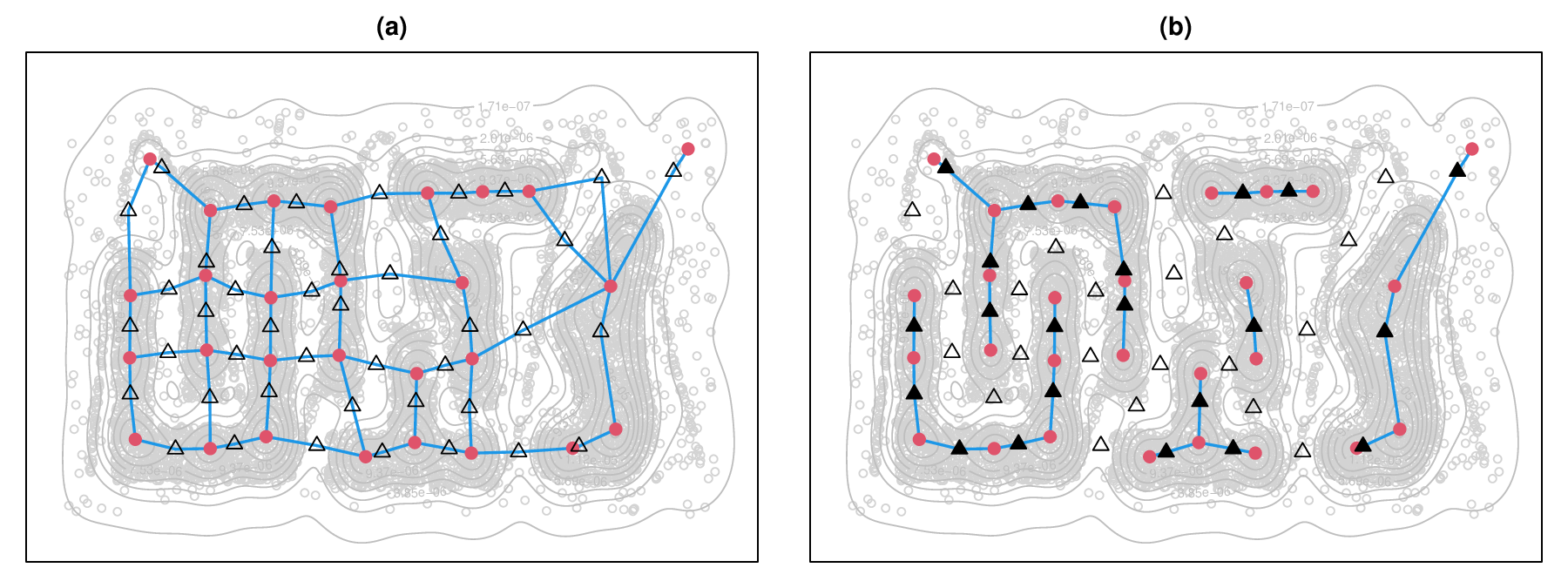}
  \caption{(a) Estimated density (contour lines) and the modal cluster
    adjacency net (line segments) for T4.8k data. The estimated
    density has 28 modes (solid points) and 43 saddle points
    (triangles). All modes are connected through the saddle points on
    the borders. (b) The 6-cluster result, yielded by MCM through
    merging initial 28 modal clusters. The saddle points used for
    merging are shown in solid triangles.}
  \label{fig:simu3-density}
\end{figure}

In some detail, MCM started with an estimated density with $28$ modes
and subsequently found $43$ saddle points by the search method
described in Section~\ref{sec:critical-points}. The mutual connection
among all modes is established after iteration 67 when the 30th saddle
point is found, and the last, 43rd one is found in iteration
360. After that, Algorithm~\ref{alg:mcm} is applied to merge every two
clusters at each step that results in the smallest loss. Unlike
DBSCAN, it does not treat any data point as noise. With MCM, these
data points form small clusters and are naturally absorbed into the
major clusters nearby. Fig.~\ref{fig:simu3-density}(a) displays the
data points, the contour lines of the estimated density, and the
initial cluster adjacency network. The darker regions formed by dense
data points align extremely well with the estimated high density
areas, both corresponding to the irregularly shaped major
clusters. This demonstrates the efficacy of density-based clustering
approaches. Fig.~\ref{fig:simu3-density}(b) shows the 6-cluster
configuration, detailing the saddle points (solid triangles) utilized
during the hierarchical merging of the initial modal clusters.

\subsection{Wireless indoor localization data} 
\label{sec:simul4}


\makebox[\linewidth][s]{The wireless indoor localization data,
  obtained from the UCI machine learning repository}
\\(\texttt{http://archive.ics.uci.edu}), contains $2,000$ observations
with $7$ numeric variables representing the signal strengths of seven
WiFi signals collected in indoor space of four rooms, and one
categorical variable denoting the room number of each record. We apply
the clustering methods to all $7$ numeric variables and study their
performance at the $4$-cluster level as compared against the known
room number information. In this scenario, the default partition
produced by PdfCluster has fewer clusters than $4$; we thus decreased
the bandwidth value of the KDE used by the clustering function to
obtain the desired number of clusters. Further, the estimated density
by MCM has $11$ modes (or critical points in $M_7$) and $13$ critical
points in $M_6$, with the last one being found by the search algorithm
in iteration $83$.

\begin{table}[!tbh]
  \centering
  \begin{tabular}{c|cccc|cccc|cccc|ccccc} \hline
    Room & \multicolumn{4}{c|}{K-means} &
                                          \multicolumn{4}{c|}{SpecClust} & \multicolumn{4}{c|}{KK-means} & \multicolumn{5}{c}{DBSCAN(20)} \\
         & 1 & 2 & 3 & 4 & 1 & 2 & 3 & 4 & 1 & 2 & 3 & 4 & 1 & 2 & 3 & 4 & 0\\
    \hline
    1 & 444 & 0 & 55 & 1 & 497 & 2 & 1 & 0   & 437 & 6 & 56 & 1& 457 & 0 & 2 & 0 & 41 \\
    2 & 8 & 447 & 45 & 0 & 0 & 498 & 2 & 0   & 0 & 472 & 28 & 0& 0 & 255 & 41 & 0 & 204 \\
    3 & 25 & 1 & 467 & 7 & 5 & 476 & 15 & 4  & 6 & 62 & 425 & 7& 3 & 0 & 433 & 2 & 62 \\
    4 & 0 & 0 & 5 & 495  & 1 & 1 & 3 & 495   & 0 & 2 & 3 & 495 & 2 & 0 & 3 & 458 & 37 \\
    \hline
    \multicolumn1l{ARI} & \multicolumn{4}{|c}{$0.816$} & \multicolumn{4}{|c}{$0.697$} & \multicolumn{4}{|c}{$0.793$} & \multicolumn{5}{|c}{ $0.721$} \\
    \hline \hline
    Room & \multicolumn{4}{c|}{PdfCluster} & \multicolumn{4}{c|}{HMAC} & \multicolumn{4}{c|}{MDE-MF} & \multicolumn{4}{c}{MCM} \\	
         & 1 & 2 & 3 & 4 & 1 & 2 & 3 & 4 & 1 & 2 & 3 & 4 & 1 & 2 & 3 & 4 \\ \hline
    1 & 495 & 3 & 0 & 2 & 498 & 1 & 0 & 1 & 465 & 0 & 32 & 3 & 499 & 0 & 1 & 0 \\
    2 & 3 & 419 & 78 & 0 & 0 & 498 & 2 & 0 & 0 & 434 & 66 & 0 & 0 & 452 & 48 & 0 \\
    3 & 10 & 481 & 0 & 9 & 5 & 489 & 0 & 6 & 1 & 0 & 495 & 4 & 4 & 0 & 490 & 6 \\ 
    4 & 1 & 2 & 0 & 497 & 1 & 1 & 0 & 498 & 0 & 0 & 6 & 494 & 2 & 0 & 0 & 498 \\ \hline
    \multicolumn1l{ARI} & \multicolumn{4}{|c}{$0.681$} & \multicolumn{4}{|c}{$0.701$} & \multicolumn{4}{|c}{$0.856$} & \multicolumn{4}{|c}{ $\mathbf{0.922}$} \\ \hline
  \end{tabular}
  \caption{Results of the clustering methods at the $4$-cluster level
    for wireless indoor localization data. Default argument settings
    are used for SpecClust and KK-means. Out of all feasible values,
    \texttt{minPts=20} gives the largest mean ARI value for DBSCAN.}
  \label{tab:simul4result}
\end{table}

For each method, a confusion table (true partition vs.\@ estimated
allocations) and the ARI values of the clustering methods at the
$4$-cluster level is given in Table~\ref{tab:simul4result}. For easy
visual inspection, the estimated clusters are re-labeled to align
reasonably well with the original room numbers.  The proposed MCM
allocates correctly the largest number of observations and therefore
has the highest ARI value. MDE-MF is not too far behind. Both
SpecClust and KK-means produce different results for different runs,
so their confusion tables are chosen from a case that that has an ARI
value closest to the averages over 50 simulation runs. With the tuning
parameter well chosen, DBSCAN has a reasonably good performance, in
terms of its identified clusters, but it labels more than 300
observations as noise (cluster 0). Both HMAC and PdfCluster have a very
small cluster containing possibly some outlying observations at the
$4$-cluster level. The clustering results that HMAC produces has one
with $4$ large clusters at the $7$-cluster level. Its ARI value is
$0.863$, which is still smaller than that of MCM. Dealing with small
clusters often requires some additional treatment in a clustering
algorithm. However, MCM handles such issues automatically, based on
its significance assessment of these clusters. The corresponding
dendrogram (omitted for brevity) further demonstrates that
partitioning the dataset into four clusters is highly appropriate.

Because MCM actually generates an explicit clustering model that
partitions the full space of $\Real^d$, it can assign new observations
to discovered clusters and predict their cluster labels. To evaluate
this predictive capability, we randomly partitioned the 2,000
observations of the dataset into a training set of 200 observations
and a test set of 1,800 observations. Over 50 independent trials, the
average ARI value on the test sets reached $0.892$ with a standard
error of $0.0095$. This performance exceeds all competing clustering
methods listed in Table~\ref{tab:simul4result}. Also notice that using
fewer observations in training tends to produce a simpler, smoother
density estimate that typically has fewer minor modes yet preserving
major ones, thus further reducing computational cost.  Consequently,
this generalization capability provides a significant advantage,
dramatically reducing computational costs for large-scale datasets.

\subsection{Olive oil data}
\label{sec:simul5}


The olive oil data \citep{forina-armanino-etal-1983} consists of $572$
observations with $8$ numeric variables measuring a series of chemical
features and $2$ categorical variables indicating the the areas of the
olive oils.  Here, we use the categorical variable \texttt{macro.area}
for group membership, which has $3$ unique values (i.e., South,
Sardinia and Centre.North). Each clustering method is performed on all
$8$ numerical variables and their results at the $3$-cluster level are
then compared.

\begin{table}[!tbh]
  \centering
  \begin{tabular}{l|ccc|ccc|ccc|cccc} \hline
    Region & \multicolumn{3}{c|}{K-means} & \multicolumn{3}{c|}{SpecClust} & \multicolumn{3}{c|}{KK-means} & \multicolumn{4}{c}{DBSCAN(5)} \\	
    & 1 & 2 & 3 & 1 & 2 & 3 & 1 & 2 & 3 & 1 & 2 & 3 & 0 \\ \hline
    \textit{South} & 219 & 102 & 2       & 317 & 0 & 6  & 200 & 113 & 10 & 306 & 0 & 0 & 17  \\
    \textit{Sardinia} & 0 & 98 & 0       & 0 & 98 & 0   & 0 & 0 & 98     & 0 & 97 & 0 & 1    \\
    \textit{Centre.North} & 0 & 30 & 121 & 0 & 23 & 128 & 0 & 0 & 151    & 0 & 0 & 147 & 4   \\ \hline
    \multicolumn1l{ARI} & \multicolumn{3}{|c}{$0.448$} & \multicolumn{3}{|c}{$0.901$} & \multicolumn{3}{|c}{$0.448$} & \multicolumn{4}{|c}{ $0.924$} \\
    \hline \hline
    Region & \multicolumn{3}{c|}{PdfCluster} & \multicolumn{3}{c|}{HMAC} & \multicolumn{3}{c|}{MDE-MF} & \multicolumn{3}{c}{MCM} \\
           & 1 & 2 & 3 & 1 & 2 & 3 & 1 & 2 & 3 & 1 & 2 & 3 \\ \hline
    \textit{South}  & 294 & 0 & 29 & 323 & 0 & 0 & 323 & 0 & 0 & 323 & 0 & 0 \\
    \textit{Sardinia}  & 0 & 98 & 0 & 0 & 98 & 0 & 0 & 98 & 0 & 0 & 98 & 0  \\
    \textit{Centre.North}  & 0& 5 & 146 & 0 & 97 & 54 & 0 & 0 & 151 & 0 & 0 & 151 \\ \hline
    \multicolumn1l{ARI} & \multicolumn{3}{|c}{$0.822$} & \multicolumn{3}{|c}{$0.816$} & \multicolumn{3}{|c}{$\mathbf{1.000}$} & \multicolumn{3}{|c}{$\mathbf{1.000}$} \\
    \hline
  \end{tabular}
  \caption{Results of the clustering methods at the $3$-cluster level
    for olive oil data. Default argument settings are used for
    SpecClust and KK-means. Out of all feasible values,
    \texttt{minPts=5} gives the largest mean ARI value for DBSCAN.}
  \label{tab:simul5result}
\end{table}

Table~\ref{tab:simul5result} shows the results of the clustering
methods at the $3$-cluster level.  We can see that both MCM and MDE-MF
divide observations into $3$ clusters that match perfectly with the
true partition. PdfCluster and HMAC incorrectly allocate a significant
proportion of observations. The K-means and KK-means methods have the
worst performance for this dataset, perhaps due to the non-spherical
irregular shapes of the clusters.  DBSCAN, with \texttt{minPts} set
optimally to 5, gives its best performance. While it provides almost a
perfect match with the region labels, it excludes 22 observations as
noise. The average performance of SpecClust is pretty good, too.

\begin{figure}[!tbh]  
  \centering
  \includegraphics[width=0.9\linewidth]{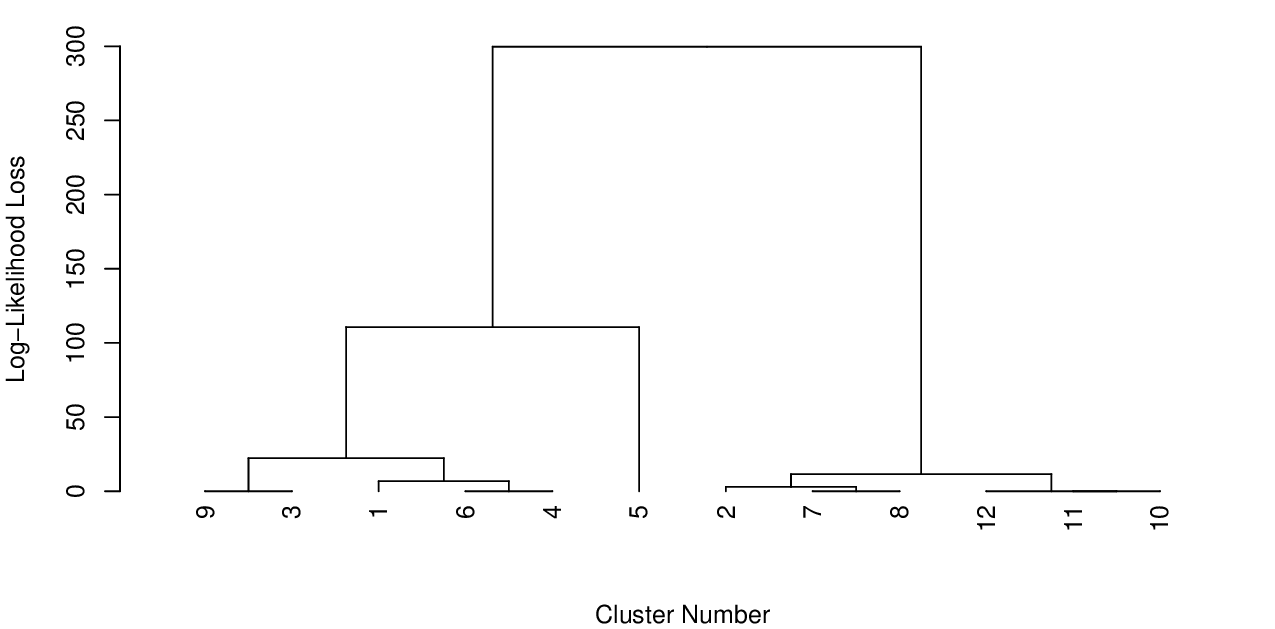}
  \caption{Dendrogram of the log-likelihood loss at each cluster-merging step
    for olive oil data}
  \label{fig:simu5_dendrogram}
\end{figure}

The dendrogram in Fig.~\ref{fig:simu5_dendrogram} suggests that
$3$-cluster is an appropriate partition, as the log-likelihood loss is
large when the number of clusters reduces from $3$ to $2$.

\section{Concluding Remarks}    
\label{sec:conclusion} 

In the above, we propose a new modal clustering method abbreviated as
MCM. It builds a hierarchical clustering structure by merging a
cluster into another until there remains only one cluster that covers
the entire feature space.  The fundamental idea is to assess the
relative significance of a cluster and then merge the least
significant cluster into one of its adjacent neighbours. It is made
feasible by exploring the properties of a Morse function that are
established in this paper for clustering purposes. In particular, one
only needs to consider the critical points of Morse index $d-1$ to
study the adjacency of clusters and assess the relative significance
of a cluster. Unlike a level-set-based method, MCM merges a cluster
into another regardless of the ``level'' of the cluster, but based on
its significance assessment using the Kullback-Leibler divergence or
log-likelihood loss. It truncates the relative bump of a cluster
without refilling the excess mass elsewhere and thus avoids density
distortion. The performance of MCM is numerically studied and compared
with several other clustering methods using both simulated and
real-world data, with very satisfactory results obtained.

Table~\ref{tab:running-time} presents the running times of our MCM
implementation, evaluated on one dataset from each of the five
scenarios studied in Sections~\ref{sec:simul1}-\ref{sec:simul5}, on a
laptop with a CPU clock speed of 4.42 GHz. The four stages of the
computation listed in the table correspond to, respectively, density
estimation (Section~\ref{sec:density-estimation}), mode identification
\& observation association (Section~\ref{sec:mode-identification}),
finding critical points in $M_{d-1}$ \& establishing cluster adjacency
(Sections~\ref{sec:critical-points} and \ref{sec:adjacency}), and
cluster merging (Sections~\ref{sec:method} and
\ref{sec:clustersignificance}). In particular, the modified NR method
keeps running and searching for new critical points in $M_{d-1}$ until
all of the following three conditions are satisfied: (a) reaching at
least 200 runs; (b) all modes mutually connected via the discovered
critical points in $M_{d-1}$; (c) having done as many runs in addition
since the last critical point is found. 

\begin{table}[!tbh]
  \centering
  \begin{tabular}{c|cccc|c} \hline Dataset & Density & Mode Id &
    $M_{d-1}$ & Merging & Total \\ \hline
    Quadrimodal & 2.54 & 0.84 & 0.64 & 0.00 & 4.01 \\
    Three circles & 4.23 & 0.21 & 1.98 & 0.71 & 7.13  \\
    T4.8k & 67.91 & 10.05 & 16.19  & 2.67 & 96.82 \\
    Wireless & 5.56  & 0.62 & 6.47 & 0.06 & 12.71 \\
    Olive oil & 5.78 & 0.25 & 3.30 & 0.05 & 9.39 \\ \hline
  \end{tabular}
  \caption{Running times (in seconds) for one dataset from each of the 5
    scenarios studied in Sections~\ref{sec:simul1}-\ref{sec:simul5}}
  \label{tab:running-time}
\end{table}

Several issues are worth discussing. First, we note that the MCM
method, as described in Section~\ref{sec:method}, can be implemented
using almost any nonparametric density estimator.  We use a
semiparametric-mixture-based density estimator owing to its better
performance than the KDE and its faster computation in a
high-dimensional space. Second, the method offers a complete hierarchy
of how initial clusters are merged in steps. This gives a user the
opportunity to choose the clustering result at a desired number of
clusters. The likelihood loss provides some evidence about the
appropriate number of clusters if one only wants to make statistically
insignificant clusters disappear. An alternative way to determine a
proper number of clusters is to conduct a significance test on the
initial modal clusters; see, e.g., \cite{burman-polonik-2009},
\cite{genovese-peronepacifico-etal-2016} and the references therein.
It may also be possible to resort to bootstrapping or cross-validation
to identify the largest number of statistically significant
clusters. Third, we use the MEM to identify modes for its nice
properties, while the gradient flow lines should really be determined
by the steepest ascent method with an infinitesimal step
size. However, in practice it does not seem to make much
difference. Fig.~\ref{fig:modalem} shows the modal clusters determined
by using the MEM iterative formula (\ref{eqn:mem}) using a very small
step size. One can hardly tell any difference between
Fig.~\ref{fig:modalem}(a) and Fig.~\ref{fig:bound}(a). The blowup in
Fig.~\ref{fig:modalem}(b) shows some slight differences, but only in
very low density areas where observations are scarce. The two
boundaries generated by the MEM are shifted a little downwards, the
strip between the two boundaries generated by the MEM is a bit wider
and the boundaries are not exactly orthogonal to the contour lines,
whereas the actual boundaries corresponding to flowlines must
be. Fourth, it may be practically advantageous to include critical
points of lower Morse indices. Allowing regions with any number of
positive eigenvalues result in finding these lower-index points
naturally while searching for points in $M_{d-1}$. They can serve as a
safeguard if some critical points in $M_{d-1}$ are undetected, as the
density heights at the points in $M_{d-2}$ are typically not too much
lower than their relevant undetected points in $M_{d-1}$.

\begin{figure}[!tbh]
  \centering
  \includegraphics[width=0.9\linewidth]{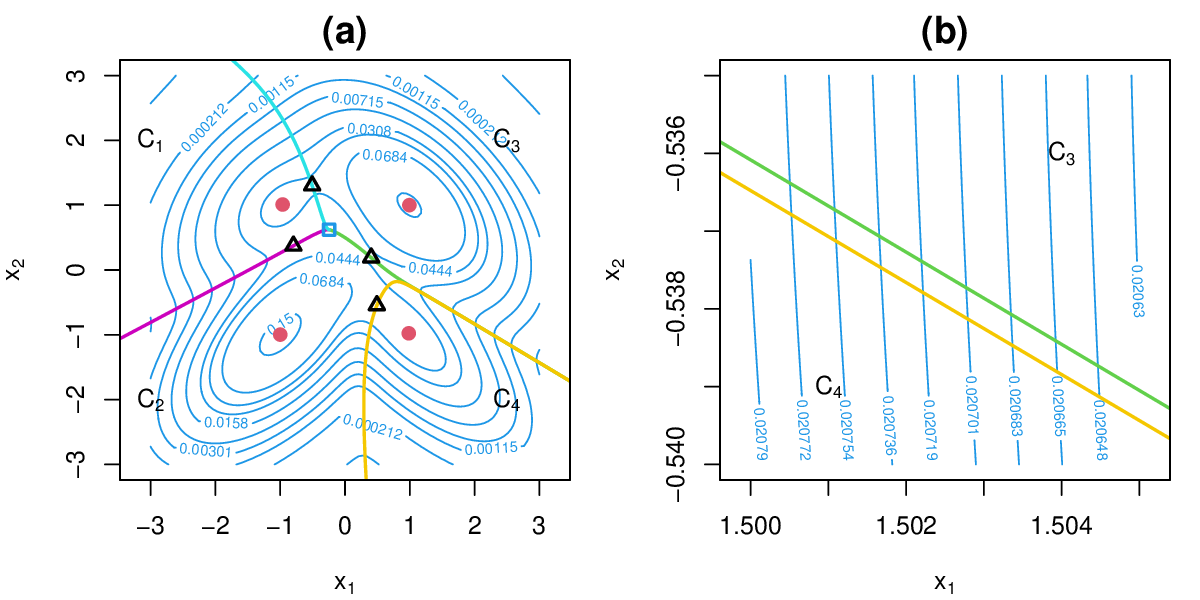}
  \caption{(a) Modal clustering based on the MEM algorithm; (b) A
    blowup of a small area}
  \label{fig:modalem}
\end{figure}

Despite its demonstrated success, the proposed method has certain
limitations. First, it is inapplicable when density estimation is
inappropriate, such as with discrete data. Second, while highly
effective at identifying major clusters, the method becomes
computationally impractical when the density estimate features too
many modes. Finally, our current implementation based on the
Newton-Raphson method is infeasible in very high-dimensional settings;
in such cases, integrating dimensionality reduction techniques like
Principal Component Analysis (PCA) can be considered.

Although this work focuses on densities with support $\Real^d$, the
proposed method can be extended to other types of support. Crucially,
the core methodology remains unchanged if the support of $f$ is a
manifold without boundary, such as the surface of a hypersphere in
directional data analysis. However, extending the framework to
manifolds with boundaries presents a greater challenge, as the maximum
density point on a boundary is not necessarily a critical
point. Adapting and implementing the method for these scenarios
remains an important direction for future research.

\small

\bibliographystyle{chicago}

\end{document}